%% file: IncDecLCS-new.tex
\documentclass[12pt]{article}

\usepackage{geometry}
\newcommand{\astring}{A}
\newcommand{\bstring}{B}

\newcommand{\difftables}{C}

\newcommand{\psam}{\mathcal{J}}
\newcommand{\ssam}{\mathcal{K}}
\newcommand{\atam}{\mathcal{L}}
\newcommand{\submat}[5]{#1[#2..#3,#4..#5]}

\newcommand{\nextmatch}[3]{\mathrm{NextMatch}_#3(#1,#2)}
\newcommand{\prevmatch}[3]{\mathrm{PrevMatch}_#3(#1, #2)}
\newcommand{\nextmatchb}[1]{\mathrm{NextMatch}_#1}
\newcommand{\prevmatchb}[1]{\mathrm{PrevMatch}_#1}

\newcommand{\nexttable}[3]{\mathrm{NextTable}_#3[#1,#2]}
\newcommand{\prevtable}[3]{\mathrm{PrevTable}_#3[#1,#2]}
\newcommand{\nexttableb}[1]{\mathrm{NextTable}_#1}
\newcommand{\prevtableb}[1]{\mathrm{PrevTable}_#1}

\newcommand{\arrayName}{BA}
\newcommand{\arrayAt}[1]{\arrayName[#1]}

\newcommand{\density}[1]{#1^\square}

\usepackage[table,xcdraw]{xcolor}
\usepackage{subfigure}
\usepackage{graphics}

\usepackage{tikz}

\usepackage{tkz-graph}
\tikzstyle{vertex}=[circle, draw, inner sep=0pt, minimum size=6pt]
\newcommand{\vertex}{\node[vertex]}

\usepackage{comment}

\usepackage{enumerate}
\usepackage{amsthm}
\usepackage{amssymb}
\usepackage{amsmath}
\theoremstyle{definition}
\newtheorem{theorem}{Theorem}[section]

\newtheorem{lemma}[theorem]{Lemma}

\newtheorem{proposition}[theorem]{Proposition}
\newtheorem{corollary}[theorem]{Corollary}

\newcommand{\RN}[1]{%
  \textup{\uppercase\expandafter{\romannumeral#1}}%
}
\newcommand{\rn}[1]{%
  \textup{\expandafter{\romannumeral#1}}%
}

\begin{document}

%\begin{frontmatter}

\title{Dynamic all scores matrices for LCS score}
\author{Amir Carmel\thanks{Department of Computer Science, Ben-Gurion University of the Negev.
%Email: \texttt{dekelts@cs.bgu.ac.il}
}
 \and Dekel Tsur$^*$ \and Michal Ziv-Ukelson$^*$}
\date{}
\maketitle

%\author[bgu]{Amir Carmel\corref{mycorrespondingauthor}}
%\ead{karmela@cs.bgu.ac.il}

%\author[bgu]{Dekel Tsur}
%\ead{dekelts@cs.bgu.ac.il}

%\author[bgu]{Michal Ziv-Ukelson}
%\ead{michaluz@cs.bgu.ac.il}

%\address[bgu]{Department of Computer Science, Ben-Gurion University of the
%Negev, Beer Sheva 84105, Israel}

\begin{comment}

\fntext[myfootnote]{Since 1880.}

\author[mymainaddress,mysecondaryaddress]{Elsevier Inc}
\ead[url]{www.elsevier.com}

\author[mysecondaryaddress]{Global Customer Service\corref{mycorrespondingauthor}}
\cortext[mycorrespondingauthor]{Corresponding author}
\ead{support@elsevier.com}

\address[mymainaddress]{1600 John F Kennedy Boulevard, Philadelphia}
\address[mysecondaryaddress]{360 Park Avenue South, New York}
\end{comment}

\begin{abstract}
The problem of aligning two strings $\astring, \bstring$
%of lengths $m,n$, respectively,
in order to determine their similarity is fundamental in the field of
pattern matching. An important concept in this domain is the ``all scores
matrix'' that encodes the local alignment comparison of two strings. Namely,
let $\ssam$ denote the all scores matrix containing the alignment score of
every substring of $\bstring$ with $\astring$, and let $\psam$ denote the all
scores matrix containing the alignment score of every suffix of $\bstring$ with every prefix of $\astring$.

In this paper we consider the problem of maintaining an all scores matrix  where
the scoring function is the LCS score, while supporting single character prepend
and append operations to $\astring$ and $\bstring$.
Our algorithms
exploit the sparsity parameters $L=LCS(\astring,\bstring)$ and $\Delta = |\bstring|-L$.
For the matrix $\ssam$ we propose
an algorithm that supports incremental operations to both ends of
$\astring$ in $O(\Delta)$ time. Whilst for the matrix $\psam$
we propose an algorithm that supports a single type
of incremental operation, either a prepend operation to $\astring$ or an append
operation to $\bstring$, in $O(L)$ time. This structure can also be extended
to support both operations simultaneously in $O(L \log \log L)$ time.
\end{abstract}

% \keywords{Sequence alignment, Longest common subsequence,
% Incremental string comparison, DIST
% matrices, All path score computations}

%\begin{keyword}
%Sequence alignment\sep Longest common subsequence\sep
%Incremental string comparison\sep DIST
%matrices\sep All path score computations
%\end{keyword}
%
%\end{frontmatter}

\section{Introduction}

In the classical problem of sequence alignment, two strings are aligned
according to some predifined scoring function. The common scoring schemes are
the Edit Distance (ED) and the Longest Common Subsequence (LCS). This problem is
fundamental in the field of pattern matching. It has broad applications to many
different fields in computer science, among others:
computer-vision, bioinformatics and natural language processing. 
Hence, it is of no suprise that this problem has attracted a vast amount of research
and publications over the years.

Given two strings $\astring$ and $\bstring$ of lengths
$m$ and $n$, respectively.
The alignment problem of $\astring$ and $\bstring$ can be naturally viewed as a
shortest path problem on the \emph{alignment graph} of $\astring$ and
$\bstring$. That is, an $(n+1) \times (m+1)$ grid graph, in which, horizontal (respectively,
vertical) edges correspond to alignment of a character in
$\astring$ (respectively, $\bstring$) with a gap, and diagonal edges correspond
to alignment of two characters in $\astring$ and $\bstring$
(see Figure~\ref{figure:k-iams}).

Landau et al.~\cite{landau1998incremental} introduced the problem of incremental
string comparison, i.e.\ given an encoding of the global comparison of two
strings $\astring$ and $\bstring$, to efficiently compute the answer for
$\astring$ and $\sigma\bstring$, and the answer for $\sigma\astring$ and 
$\bstring$.
Furthermore, they show how incremental string comparison can be used to obtain
more efficient algorithms for various problems in pattern matching, such as:
the longest prefix approximate matching problem, the approximate overlap problem
and cyclic string comparison.
Incremental string comparison has also been applied to finding all approximate
gapped palindromes~\cite{hsu2010finding}, approximate regularities in
strings~\cite{christodoulakis2005implementing, zhang2008generalized},
consecutive suffix alignment~\cite{hyyro2008efficient, landau2007two} and
more~\cite{lai2017online, sokol2014speeding}. 
Several
improvements to the algorithm of Landau et al.~\cite{landau1998incremental} have been proposed
over the years~\cite{hyyro2016compacting, hyyro2015dynamic, ishida2005fully,
kim2000dynamic}.

All scores matrices were introduced by Apostolico et
al.~\cite{ApostolicoALM90} in order to obtain fast parallel algorithms for LCS computation.
%{\color{red}The \emph{all scores matrix} is a matrix that stores the optimal
%alignment scores between:
%(\RN{1}) $\astring$ against every substring of $\bstring$, (\RN{2}) $\bstring$ against
%every substring of $\astring$, (\RN{3}) every suffix of $\astring$ against every
%prefix of $\bstring$, and (\RN{4}) every suffix of $\bstring$ against every prefix of
%$\astring$.
%The all scores matrix can be divided into four blocks where each block stores
%the scores of alignments of only one type from the four types above.
%Each block will be called a \emph{partial all scores matrix}.
%We denote by $\ssam$ the partial all scores matrix containing the
%optimal alignment scores of type~(\RN{2}) above, and
%denote by $\psam$ the partial all scores matrix containing the optimal alignment scores
%of type~(\RN{4}).
%All scores matrices are also called DIST
%matrices~\cite{ApostolicoALM90, Schmidt98} or highest-score
%matrices~\cite{tiskin2008semilocal}.
%}
%
An \emph{all scores matrix} is a matrix that stores the optimal
alignment scores of one or more types from the following types:
(\RN{1}) $\bstring$ against every substring of $\astring$,
(\RN{2}) $\astring$ against every substring of $\bstring$,
(\RN{3}) every suffix of $\astring$ against every
prefix of $\bstring$,
and
(\RN{4}) every suffix of $\bstring$ against every prefix of $\astring$.
We denote by $\atam$ the all scores matrix containing the optimal alignment scores
of all the types defined above.
We denote by $\ssam$ the all scores matrix containing the
optimal alignment scores of type~(\RN{2}), and
denote by $\psam$ the all scores matrix containing the optimal alignment scores
of type~(\RN{4})
(see Figure~\ref{figure:k-iams}).
Due to symmetry, we will ignore the all scores matrices that contain optimal alignment scores
of type~(\RN{1}) or of type~(\RN{3}).
All scores matrices are also called DIST
matrices~\cite{ApostolicoALM90, Schmidt98} or highest-score
matrices~\cite{tiskin2008semilocal}.

The problem of constructing all scores matrices has been studied in~\cite{alves2008all, Carmel2018, MatarazzoTZ,
sakai2018substring, Schmidt98,tiskin2008semilocal, tiskin2008semi}.
We note that this problem is a special case of the more general problem of
computing all pairs shortest paths in planar graphs~\cite{cabello2006many, cabello2013multiple, cohen2017fast, eisenstat2013linear, gawrychowski2018better, mozes2012exact}.

For an $n \times n$ matrix $\mathcal{D}$, its \emph{density}
matrix $\density{\mathcal{D}}$ is an $(n-1) \times (n-1)$ matrix, where
$\density{\mathcal{D}}[i,j]=\mathcal{D}[i-1,j-1]+\mathcal{D}[i,j]-\mathcal{D}[i-1,j]-\mathcal{D}[i,j-1]$.
A matrix is called \emph{Monge} (resp. \emph{anti-Monge}) if its density matrix
is non-negative (resp. non-positive), and \emph{sub-unit Monge} (resp.
\emph{sub-unit anti-Monge}) if every row or column of the density matrix
contains at most one non-zero element, and all the non-zero elements are equal
to $1$ (resp. $-1$). All scores matrices are known to be sub-unit
Monge or sub-unit anti-Monge~\cite{tiskin2008semilocal}. Hence, all scores
matrices can be encoded in linear space. Consequently, supporting incremental operations for all scores matrices can also serve to further reduce the space complexity for
Incremental String Comparison, which was noted in~\cite{hyyro2015dynamic} to be
the main practical limitation.

%\textcolor{blue}{
%In this paper we propose a variant of the classical
%incremental string comparison problem, denoted \emph{$k$-sided Incremental All
%Scores Matrices ($k$-IASM)}. In this problem variant, we consider data
%structures that encode all scores matrices and support a specific collection of size $k$ of prepend and append operations to
%$\astring$ and $\bstring$. 
%}

We next consider the problem of incremental construction of all
scores matrices.
That is, we wish to maintain an all scores matrix of two strings $\astring$
and $\bstring$ while supporting all or some of the following operations:
appending a character to $\astring$, prepending a character to $\astring$,
appending a character to $\bstring$, and prepending a character to $\bstring$.
If the number of supported operations from the four operations above is $k$,
we refer to the problem as the \emph{$k$-sided incremental all scores matrix
problem ($k$-IASM)}.

The algorithm of Schmidt~\cite{Schmidt98} for all scores matrix construction also
solves the $2$-IASM problem on the matrix $\atam$.
For discrete score functions, including LCS scores,
the algorithm of Schmidt requires $O(m+n)$ space and it
supports incremental operations to the right ends of $\astring$ and
$\bstring$ in $O(n)$ and $O(m)$ time, respectively.
%This algorithm was later
%adapted by Alves et al.~\cite{alves2008all} to apply for the LCS metric.
For LCS scores, Tiskin~\cite[Chapter 5.3]{Tiskin} utilized the property that the
all scores matrix $\atam$ is unit-Monge.
His algorithm solves the $4$-IASM problem on $\atam$. The algorithm uses $O(m+n)$ space and
supports incremental operations in either $O(n)$ or $O(m)$ time.
Restricting the set of supported incremental operations and
maintaining either the matrix $\psam$ or $\ssam$ rather than the full matrix $\atam$
allows for faster algorithms.
Such an algorithm was given in Landau et al.~\cite{landau2007two}.
The algorithm of Landau et al.\ allows prepend operations to the string
$\astring$ while maintaining the alignment score between every suffix of
$\bstring$ to every prefix of $\astring$, thus it solves $1$-IASM on the matrix
$\psam$.
Amortized on $m$ incremental
operations, the algorithm runs in $O(L)$ time per operation,
where $L$ is the length of the longest common subsequence of $\astring$ and $\bstring$.
The space complexity of the algorithm is $O(n)$.

\subsection{Our Contribution}
%This paper suggests another construction approach that allows online streaming of the data, by supporting a series of incremental extensions of the two aligned strings via character append or prepend operations.

In this paper we study the $k$-IASM problem under LCS score.
We exploit the sparsity parameters $L$ and $\Delta =n-L$, to design faster algorithms.
These parameters have already been extensively used to
give faster LCS algorithms~\cite{apostolico1987longest,
bringman2018multivariate, hirschberg1977algorithms, myers1986ano, wu1990np}.
In order to obtain faster running times, we restrict the set of supported
incremental operations (namely $k$ is either $1$ or $2$), and we also
maintain either $\psam$ or $\ssam$ and not the full matrix $\atam$.
%(see Figure~\ref{figure:k-iams}).

%In our work we consider two types of all scores matrices.
%We denote by $\ssam$ the all scores matrix containing the
%optimal alignment between $\bstring$ and every substring of $\astring$, and
%denote by $\psam$ the all scores matrix containing the optimal alignment
%between every suffix of $\bstring$ and every prefix of $\astring$. For each such
%matrix we consider a specific set of prepend/append operations that  {\color{red} coincides with COMMMENT:NOT CLEAR THAT THIS IS THE RIGHT TERM TO USE HERE, PERHAPS "GOES AGAINST THE FLOW OF"} the alignment flow of the respective all scores matrix (see
%Figure~\ref{figure:k-iams}).

\begin{figure}
\centering
\subfigure[The all substring-string LCS problem]{
\input{figure_k-iams-ssam-new}
}
% \hspace{2cm}
\quad
\subfigure[The all prefix-suffix LCS problem]{
\input{figure_k-iams-psam-new}
}
\caption{The alignment graph for the strings $\astring = ccbccaa$ and $\bstring
= cacab$. Figure~(a) illustrates a prepend operation of character $c$ to
 $\astring$. The matrix $\ssam$ contains the optimal scores of paths from every vertex
 on the left border to every vertex on the right border (these vertices are
 colored gray).
%, while the matrix $\psam'$ corresponds to the newly appended
% column.
 The optimal path that corresponds to $\ssam [2,4]$ is marked using
 thicker edges.
Figure~(b) illustrates an append operation of character $b$ to $\bstring$.
The matrix $\psam$ contains the optimal scores of paths from every vertex
on the left border to every vertex on the bottom border (these vertices are
colored gray).}
\label{figure:k-iams}
\end{figure}
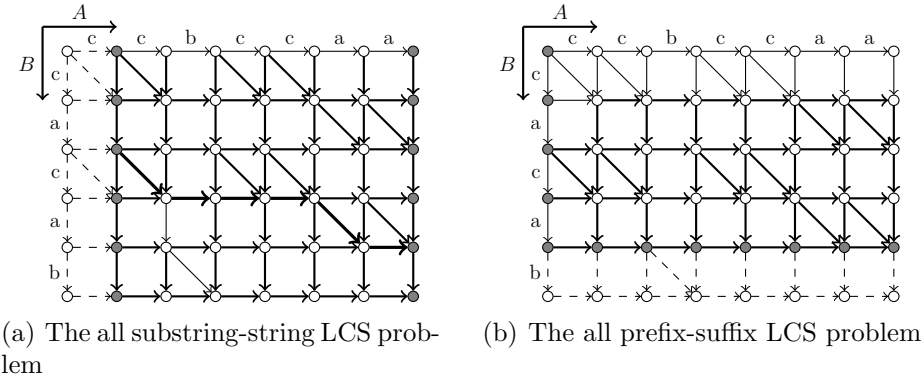

We note that various applications that utilize incremental all scores matrices
require a specific type of all scores matrix, and also use restricted sets
of update operations (see for example~\cite{hsu2010finding,hyyro2008efficient,
KentLZ, landau2007two, LandauSZ03, LandauSS01, LandauZ01, sakai2011almost,
Schmidt98, tiskin2008semi}). Thus, our algorithms are advantageous to these applications.

Our results are as follows (see also Table~\ref{tabel:iams}).
In Section~\ref{section:k-ssam} we give an algorithm for maintaining the
matrix $\ssam$ while supporting incremental operations to both sides
of string $\astring$ in $O(\Delta)$ time, using $O(n)$ space.
%{\color{red} COMMENT:SINCE YOU SORT OF GIVE A ROADMAP FOR SEC 3, IT SEEMS RIGHT TO BRIEFLY MENTION WHAT YOU DO IN THE UPCOMING SECS 2 AND 4}.
In Section~\ref{section:k-psam} we give an algorithm for maintaining
the matrix $\psam$ and supporting one type of update operations: either prepending
a character to $\astring$ or appending a character to $\bstring$.
The algorithm uses either $O(n)$ or $O(m)$ space, and the worst-case time
complexity of an update operation is $O(L)$.
We also show how to support both update operations simultaneously. This
increases the space complexity to $O(m+n)$ and the worst-case time complexity
of an update operation to $O(L \log \log L)$.
Our result improves the result of Landau et al.~\cite{landau2007two}
since the update time in Landau et al.\ is $O(L)$ amortized.
Additionally, the proof of correctness of our algorithm is simpler.
Finally, in Section~\ref{section:ssam+psam}, we give an algorithm for
maintaining both matrices $\ssam$ and $\psam$ while supporting append operations
to either string.
Here, the space complexity is $O(m+n)$ and the update time is $O(\Delta)$ or
$O(L)$.

%Table~\ref{tabel:iams}
%summarizes current results for the $k$-IASM problem and shows our proposed
%improvements.
%Moreover, our solution to this problem also has the
%advantage that it simplifies the combinatorial
%properties required for the algorithm (cf.~\cite{landau2007two, Tiskin}).
We note that our proposed approach could also be applied to yield a
solution to the 4-sided problem with the same time and space complexity as the algorithm
of Tiskin~\cite{Tiskin}.

\begin{table}[]
\centering
\footnotesize
\begin{tabular}{|l|l|l|l|l|}
\hline
                                                                & Problem         & Supported operations                                                                                               & Time               & Space        \\ \hline
Landau et al.~\cite{landau2007two}        & 1-sided $\psam$ & Prepend  a
character to $\astring$                                                                                & 
$O(L)$ amortized     & $O(n)$       \\ \hline
Schmidt~\cite{Schmidt98} & 2-sided $\atam$ & Append a character to $\astring$ or $\bstring$ & $O(n)$/$O(m)$ & $O(m+n)$\\ \hline
Tiskin~\cite{Tiskin} & 4-sided $\atam$ & Prepend/append a character to $\astring$ or $\bstring$   & $O(n)$/$O(m)$  & $O(m+n)$     \\ \hline
%Ours (Thm.~\ref{theorem:1-sided-ssam}) & {\color{red}1-sided $\ssam$} & Prepend
%or append a character to $\bstring$                                                       
%& $O(\Delta)$        & $O(\Delta)^\ast$  \\ \hline
Ours
(Thm.~\ref{theorem:2-sided-ssam}) & 2-sided $\ssam$ & Prepend and append a
character to $\astring$                                                         
& $O(\Delta)$        & $O(n)$       \\ \hline Ours
(Thm.~\ref{theorem-1-sided-psam}) & 1-sided $\psam$ &
\begin{tabular}[c]{@{}l@{}}Prepend a character to $\astring$, or \\ append a
character to $\bstring$\end{tabular} & $O(L)$             &
\begin{tabular}[c]{@{}l@{}} $O(n)$ \\ $O(m)$ \end{tabular} \\ \hline Ours
(Thm.~\ref{theorem:2-sided-psam}) & 2-sided $\psam$ & \begin{tabular}[c]{@{}l@{}}Prepend a character to $\astring$, and\\ append a character to $\bstring$\end{tabular} & $O(L \log \log L)$ & $O(m+n)$     \\ \hline
Ours (Thm.~\ref{theorem:ssam+psam}) & 2-sided $\ssam,\psam$ & Append a character to
$\astring$ or $\bstring$ & $O(\Delta)$/$O(L)$ & $O(m+n)$\\ \hline
\end{tabular}
\caption{Comparison of known and new results for $k$-IASM.
$\astring$ and
$\bstring$ denote two strings of lengths $m$ and $n$, respectively, over a
constant alphabet. In this setting, $L=LCS(\astring,\bstring)$ and
$\Delta = n-L$.\label{tabel:iams}
%In the $1$-sided $\ssam$ problem we do not take into account the complexity of
%the auxilary data-structure that allows constant time access to any position in
%$\astring$, hence we mark by $\ast$ the space complexity in this case.
%We also use the $\dagger$ symbol to denote amortized running time.
}
%}
\end{table}

% \emph{Our first result, described in detail in
% Section~\ref{section:inc-dec-lcs}, extends the work of Ishida et al.
% to also include decremental operations.}
%1. our ``confinment'' goes with the natural flow of the problem. that is, if we
% are interested in A-prefix-B-suffix alignment, then we support incremental to
% the head of A - that is the prefixes and to the tail of B - that is the
% suffixes. when we are interested in an alignment of every substring of A to
% the whole of B then we support incremental of A in both directions. so note
% that the incremental operations fit to the problem being solved.

%2. the second algorithm in ~\cite{landau2007two}, although was suggested for
% the purpose of consectuive suffix alignment, cannot be immediately extended to
% support both directions. since this algorithm utilizes a sorting of the
% encoding.

\section{Preliminaries}
We denote $[i:j] = \{i,i+1,\ldots,j\}$.
Let $\astring,\bstring$ be two strings of lengths $m,n$, respectively,
over an alphabet $\Sigma$ of constant size.
Denote by $G_{\astring,\bstring}$ the alignment graph of
$\astring$ and $\bstring$.
% We write $G$ when $\astring$ and $\bstring$ are clear from the context.
$G_{\astring,\bstring}$ is a grid graph over
a vertex set $[0:n]\times [0:m]$, where all vertical and horizontal
edges are present, and a diagonal edge between $(i-1,j-1)$ to $(i,j)$ is present if and only
if $\bstring[i]=\astring[j]$, in which case we say that $(i,j)$ is a \emph{match
point} in $G_{\astring,\bstring}$.
Diagonal edges have score 1, and horizontal and vertical edges have a score 0.
We denote $L=LCS(\astring,\bstring)$ and $\Delta = n-L$.

For a string $S=\sigma_1 \sigma_2 \cdots \sigma_n$, let $S[i..j] = \sigma_{i+1}
\cdots \sigma_j$ denote the substring of $S$ from $i+1$ to $j$. Consequently,
$S[i..i]$ denotes the empty string, for which we use the symbol $\epsilon$.
For a character $\sigma$ and a string $S$, let $\nextmatch{i}{\sigma}{S}$
denote the minimum index $i' > i$, such
that $\sigma = S[i']$ (and $\infty$ if no such $i'$ exists), and let
$\prevmatch{j}{\sigma}{S}$ denote the maximum index $j' \leq j$, such that
$\sigma = S[j']$ (and $-\infty$ if no such $j'$ exists).

We consider two types of all scores matrices (see
Figure~\ref{figure:k-iams}). Let $\psam$ denote the matrix
containing the scores of the optimal paths between every vertex on the left
border of $G_{\astring,\bstring}$ to every vertex on its bottom border, that
is, $\psam [i,j]= LCS(\bstring[i..n], \astring[0..j])$.
Also, let $\ssam$ denote the matrix containing the scores of the optimal
paths between every vertex on the left border of $G_{\astring,\bstring}$ to
every vertex on its right border, namely $\ssam [i,j]= LCS(\bstring[i..j],
\astring)$. If $i>j$ then no such path exists and we define $\ssam [i,j] = j-i$. 

For a matrix $\mathcal{D}$ over index set $[0:n] \times [0:m]$, 
we define its \emph{density matrix}, denoted by $\density{\mathcal{D}}$, 
to be a matrix over the index set $[1:n] \times [1:m]$ such that
$\density{\mathcal{D}}[i,j] =
(\mathcal{D}[i,j]+\mathcal{D}[i-1,j-1])-(\mathcal{D}[i-1,j]+\mathcal{D}[i,j-1])$.
The next property follows immediately from the above definition.

\begin{proposition}\label{proposition:density}
$\mathcal{D}[i,j] = \sum\limits_{\substack{1 \leq i' \leq i \\ 1 \leq j' \leq
j}}\density{\mathcal{D}}[i',j'] - \mathcal{D}[0,0] + \mathcal{D}[0,j] +
\mathcal{D}[i,0], $
for every $0 \leq i \leq n$ and $0 \leq j \leq m$.
\end{proposition}

We say that a matrix is \emph{sub-unit Monge} (resp., \emph{sub-unit anti-Monge}), if every row
and column of its density matrix contains at most one non-zero element, and all the non-zero elements
are equal to $1$ (resp., $-1$). It is well established that the matrices $\ssam$
and $\psam$ are sub-unit Monge and sub-unit anti-Monge, respectively
(cf.~\cite{tiskin2008semilocal}).
This implies the next corollary regarding the size of the encoding of $\ssam$ and $\psam$.

\begin{corollary}\label{corollary:encoding-dist}
The density matrix of $\psam$ has exactly $L$ non-zero elements,
and the density matrix of $\ssam$ has exactly $\Delta$ non-zero elements.
\end{corollary}

\begin{proof}
For the first part, note that $\psam [n,m] = LCS(\bstring [n..n], \astring) =
0$.
However, by Proposition~\ref{proposition:density}, $\psam [n,m] =
\sum\limits_{\substack{1 \leq i' \leq n \\ 1 \leq j' \leq
m}}\density{\psam}[i',j'] - \psam[0,0] + \psam[0,m] + \psam[n,0]$. Hence,
$\psam[0,m] + \psam[n,0] - \psam[0,0] = -\sum\limits_{\substack{1 \leq i' \leq n \\
1 \leq j' \leq m}}\density{\psam}[i',j']$. 
By the definition of the matrix $\psam$, $\psam [0,0] = LCS(\bstring,\epsilon) =
0$, $\psam [0,m] = LCS(\bstring,\astring)=L$ and $\psam [n,0] = LCS(\epsilon,
\epsilon) = 0$. %, where $\epsilon$ denotes the empty string.
Therefore, $\sum\limits_{\substack{1 \leq i' \leq n \\ 1 \leq j' \leq
m}}-\density{\psam}[i',j'] = L$.
The second part of the corollary can be obtained similarly,
by the definition of the matrix $\ssam$.
\end{proof}

In what follows the non-zero cells of a density matrix are called
\emph{pivotal points}.
We encode an all scores matrix by storing the pivotal
points of its density matrix. This requires $O(L)$ space for
the $\psam$ matrix, and $O(\Delta)$ space for $\ssam$ (by
Corollary~\ref{corollary:encoding-dist}). In the subsequent sections we handle
the two matrices $\ssam$ and $\psam$ separately, and for each matrix we
consider both the
$1$-sided and the $2$-sided problems.

\section{Incremental $\ssam$ matrix}\label{section:k-ssam}
Recall that we need to maintain the all scores matrix, denoted $\ssam$,
of the strings $\astring$ and $\bstring$.
Our algorithm encodes the matrix $\ssam$ by storing the pivotal points of its density matrix.
Consider an incremental operation that prepends a character $\sigma$
to $\astring$, that is $\astring' = \sigma \astring$.
We need to compute the pivotal points of $\density{\ssam'}$, where $\ssam'$ is
the all scores matrix of $\astring'$ and $\bstring$.
Define $\difftables = \ssam' - \ssam$.
By definition, $\density{\ssam'} = \ssam^\square + \difftables^\square$.
We will next show how to compute the pivotal points of $\difftables^\square$.
This will allow the algorithm to compute the pivotal points of $\ssam'^\square$.
$\difftables$ contains either 0 or 1
values, since the new character $\sigma$ can increase the length of the LCS of
each alignment by at most 1.
See Figure~\ref{figure:inc-b-left-ssam} for an example of matrices $\ssam', \ssam,
\difftables$ and $\density{\difftables}$.

The optimal path from the $i$'th vertex on the first column to the
$j$'th vertex on the last column of $G_{\sigma \astring,\bstring}$ either
utilizes a match point that was generated due to the new prepended
character $\sigma$, or not. In the former case we can assume the optimal
path uses the topmost such match point, and in the latter case the score of
this path is equal to $\ssam [i,j]$.
Hence, we have $\ssam' [i,j] = \max \{
\ssam[\nextmatch{i}{\sigma}{\bstring},j]+1, \ssam[i,j] \}$.
This leads to the following proposition.
\begin{proposition}\label{proposition:left-inc-b}
$\ssam'[i,j] = \ssam[i,j] + 1 $ if and only if $\nextmatch{i}{\sigma}{\bstring}
< \infty$ and $\ssam [i,j]=\ssam [\nextmatch{i}{\sigma}{\bstring},j]$.
\end{proposition}

We now describe our main lemma required for the incremental step. 

\begin{lemma}\label{lemma:left-inc-b-main}
$\difftables [i,j]=1$ if and only if $\nextmatch{i}{\sigma}{\bstring} < \infty$
and every row of
$\submat{\ssam^\square}{i+1}{\nextmatch{i}{\sigma}{\bstring}}{1}{j}$ has exactly
one pivotal point.
%there are $\nextmatch{i}{\sigma}{\astring}-i$ pivotal points in the submatrix
%$\submat{\ssam^\square}{i+1}{\nextmatch{i}{\sigma}{\astring}}{1}{j}$.
\end{lemma}

\begin{proof}
Let $\nextmatch{i}{\sigma}{\bstring}=k$.
By Proposition~\ref{proposition:density},
$\ssam [i,j]=\ssam [k,j]$ if and only if: \\
\begin{equation*}
\sum\limits_{\substack{1
\leq i' \leq i \\ 1 \leq j' \leq j}}\ssam^{\square}[i',j'] - \ssam[0,0] +
\ssam[0,j] + \ssam[i,0] = \sum\limits_{\substack{1 \leq i' \leq k
\\
1 \leq j' \leq j}}\ssam^{\square}[i',j'] - \ssam[0,0] + \ssam[0,j] +
\ssam[k,0].
\end{equation*}
After canceling the terms that appear in both sides, we obtain the 
equality
\begin{equation*}
\sum\limits_{\substack{i+1 \leq i' \leq k  \\ 1 \leq j' \leq
j}}\ssam^{\square}[i',j'] = \ssam[i,0] - \ssam[k,0].
\end{equation*} 
By the definition of the matrix $\ssam$, $\ssam[i,0]=-i$ and
$\ssam[k,0]=-k$, thus obtaining
\begin{equation*}
\sum\limits_{\substack{i+1 \leq i' \leq k  \\ 1 \leq j' \leq
j}}\ssam^{\square}[i',j'] = k-i.
\end{equation*}
Recall that $\ssam$ is a sub-unit Monge matrix and hence there is at most one pivotal point in every row of
$\ssam^\square$.
Therefore, $\sum\limits_{\substack{i+1 \leq i' \leq k  \\ 1 \leq j' \leq
j}}\ssam^{\square}[i',j'] = k-i$
%the number of pivotal points in the submatrix
%$\submat{\ssam^\square}{i+1}{\nextmatch{i}{\sigma}{\astring}}{1}{j}$ is equal to $\nextmatch{i}{\sigma}{\astring}-i$,
if and only if every row of
$\submat{\ssam^\square}{i+1}{\nextmatch{i}{\sigma}{\bstring}}{1}{j}$ has exactly
one pivotal point.
The lemma now follows from Proposition~\ref{proposition:left-inc-b}.
\end{proof}

\begin{figure}[]
\centering
\begin{minipage}{0.4\linewidth}
\subfigure[$\ssam'$]{
\resizebox{1.3\textwidth}{!}{%
\begin{tabular}{|l|l|l|l|l|l|l|l|l|l|}
\hline
0  & 1  & 2  & 2                         & 2                         & 3  & 4  & 4                         & 4                         & 4                         \\ \hline
-1 & 0  & 1  & \cellcolor[HTML]{C0C0C0}2 & 2                         & 3  & 4  & 4                         & 4                         & 4                         \\ \hline
-2 & -1 & 0  & 1                         & \cellcolor[HTML]{C0C0C0}2 & 3  & 4  & 4                         & 4                         & 4                         \\ \hline
-3 & -2 & -1 & 0                         & 1                         & 2  & 3  & 3                         & 3                         & 3                         \\ \hline
-4 & -3 & -2 & -1                        & 0                         & 1  & 2  & \cellcolor[HTML]{C0C0C0}3 & 3                         & 3                         \\ \hline
-5 & -4 & -3 & -2                        & -1                        & 0  & 1  & 2                         & 2                         & 2                         \\ \hline
-6 & -5 & -4 & -3                        & -2                        & -1 & 0  & 1                         & \cellcolor[HTML]{C0C0C0}2 & 2                         \\ \hline
-7 & -6 & -5 & -4                        & -3                        & -2 & -1 & 0                         & 1                         & \cellcolor[HTML]{C0C0C0}2 \\ \hline
-8 & -7 & -6 & -5                        & -4                        & -3 & -2 & -1                        & 0                         & 1                         \\ \hline
-9 & -8 & -7 & -6                        & -5                        & -4 & -3 & -2                        & -1                        & 0                         \\ \hline
\end{tabular}
}
}

\subfigure[$\ssam$]{
\resizebox{1.3\textwidth}{!}{%
\begin{tabular}{|l|l|l|l|l|l|l|l|l|l|}
\hline
0  & 1  & 2  & 2                         & 2                         & 2                         & 3  & 3                         & 3                         & 3                         \\ \hline
-1 & 0  & 1  & \cellcolor[HTML]{C0C0C0}2 & 2                         & 2                         & 3  & 3                         & 3                         & 3                         \\ \hline
-2 & -1 & 0  & 1                         & 1                         & \cellcolor[HTML]{C0C0C0}2 & 3  & 3                         & 3                         & 3                         \\ \hline
-3 & -2 & -1 & 0                         & \cellcolor[HTML]{C0C0C0}1 & 2                         & 3  & 3                         & 3                         & 3                         \\ \hline
-4 & -3 & -2 & -1                        & 0                         & 1                         & 2  & 2                         & 2                         & 2                         \\ \hline
-5 & -4 & -3 & -2                        & -1                        & 0                         & 1  & \cellcolor[HTML]{C0C0C0}2 & 2                         & 2                         \\ \hline
-6 & -5 & -4 & -3                        & -2                        & -1                        & 0  & 1                         & \cellcolor[HTML]{C0C0C0}2 & 2                         \\ \hline
-7 & -6 & -5 & -4                        & -3                        & -2                        & -1 & 0                         & 1                         & 1                         \\ \hline
-8 & -7 & -6 & -5                        & -4                        & -3                        & -2 & -1                        & 0                         & \cellcolor[HTML]{C0C0C0}1 \\ \hline
-9 & -8 & -7 & -6                        & -5                        & -4                        & -3 & -2                        & -1                        & 0                         \\ \hline
\end{tabular}
}
}
\end{minipage}
\hfill
\begin{minipage}{0.4\linewidth}
% \vspace{-10pt}
\subfigure[$\difftables$]{
\resizebox{1.2\textwidth}{!}{%
\begin{tabular}{|l|l|l|l|l|l|l|l|l|l|}
\hline
0 & 0 & 0 & 0 & 0                         & \cellcolor[HTML]{E8E8E8}1 & 1 & 1                         & 1 & 1                         \\ \hline
0 & 0 & 0 & 0 & 0                         & \cellcolor[HTML]{E8E8E8}1 & 1 & 1                         & 1 & 1                         \\ \hline
0 & 0 & 0 & 0 & \cellcolor[HTML]{E8E8E8}1 & 1                         & 1 & 1                         & 1 & 1                         \\ \hline
0 & 0 & 0 & 0 & 0                         & 0                         & 0 & 0                         & 0 & 0                         \\ \hline
0 & 0 & 0 & 0 & 0                         & 0                         & 0 & \cellcolor[HTML]{E8E8E8}1 & 1 & 1                         \\ \hline
0 & 0 & 0 & 0 & 0                         & 0                         & 0 & 0                         & 0 & 0                         \\ \hline
0 & 0 & 0 & 0 & 0                         & 0                         & 0 & 0                         & 0 & 0                         \\ \hline
0 & 0 & 0 & 0 & 0                         & 0                         & 0 & 0                         & 0 & \cellcolor[HTML]{E8E8E8}1 \\ \hline
0 & 0 & 0 & 0 & 0                         & 0                         & 0 & 0                         & 0 & 0                         \\ \hline
0 & 0 & 0 & 0 & 0                         & 0                         & 0 & 0                         & 0 & 0                         \\ \hline
\end{tabular}
}
}

\subfigure[$\difftables^\square$]{
\resizebox{1.2\textwidth}{!}{%
\begin{tabular}{|l|l|l|l|l|l|l|l|l|}
\hline
0 & 0 & 0 & 0                          & 0                          & 0 & 0                          & 0 & 0                          \\ \hline
0 & 0 & 0 & \cellcolor[HTML]{C0C0C0}1  & \cellcolor[HTML]{C0C0C0}-1 & 0 & 0                          & 0 & 0                          \\ \hline
0 & 0 & 0 & \cellcolor[HTML]{C0C0C0}-1 & 0                          & 0 & 0                          & 0 & 0                          \\ \hline
0 & 0 & 0 & 0                          & 0                          & 0 & \cellcolor[HTML]{C0C0C0}1  & 0 & 0                          \\ \hline
0 & 0 & 0 & 0                          & 0                          & 0 & \cellcolor[HTML]{C0C0C0}-1 & 0 & 0                          \\ \hline
0 & 0 & 0 & 0                          & 0                          & 0 & 0                          & 0 & 0                          \\ \hline
0 & 0 & 0 & 0                          & 0                          & 0 & 0                          & 0 & \cellcolor[HTML]{C0C0C0}1  \\ \hline
0 & 0 & 0 & 0                          & 0                          & 0 & 0                          & 0 & \cellcolor[HTML]{C0C0C0}-1 \\ \hline
0 & 0 & 0 & 0                          & 0                          & 0 & 0                          & 0 & 0                          \\ \hline
\end{tabular}
}
}
\end{minipage}

\caption{An example of the matrices $\mathcal{K}, \mathcal{K'},\difftables$ and
$\difftables^\square$ for $\astring = bbcac$, $\bstring = ccabaccaa$,
and a prepend of the character $a$ to $\astring$.
Pivotal points are colored dark gray and step indices in $\difftables$ are
colored light gray.}
\label{figure:inc-b-left-ssam}
\end{figure}

The following corollary follows immediatly from Lemma~\ref{lemma:left-inc-b-main}.

\begin{corollary}\label{corollary:left-inc-b-monotonicity-diff-table}
If $\difftables [i,j] = 1$ then $\difftables [i,j'] = 1$ for every $j' \geq j$.
\end{corollary}
 
We say that $(i,j)$ is a \emph{step index} in $\difftables$ if
$\difftables[i,j] \neq \difftables[i,j-1]$
%$\difftables[i,j]$ contains the first non-zero element at row $i$ of $\difftables$
(see Figure~\ref{figure:inc-b-left-ssam}).
The following lemma shows that the pivotal points of $\difftables^\square$ can
be obtained from the step indices of $\difftables$.
Therefore, the computation of these indices will be the focus of our algorithm.

\begin{lemma}\label{lemma:pivotal-points}
For every cell $(i,j)$ in $\density{\difftables}$,
\begin{itemize}
  \item $\density{\difftables}[i,j]=1$ if and only if $(i,j)$ is a step index in
  $\difftables$ and $(i-1,j)$ is not a step index.
  \item $\density{\difftables}[i,j]=-1$ if and only if $(i,j)$ is not a step
  index in $\difftables$ and $(i-1,j)$ is a step index.
\end{itemize}
\end{lemma}

\begin{proof}
Let $\bar{\difftables}$ be a matrix containing the columns differences of the matrix
$\difftables$, that is,
$\bar{\difftables}[i,j]=\difftables[i,j]-\difftables[i,j-1]$.
From Corollary~\ref{corollary:left-inc-b-monotonicity-diff-table} we have that
$\bar{\difftables}[i,j]=1$ if and only if $(i,j)$ is a step index in
$\difftables$ (note that if $(i,j)$ is a step index then $j>0$). The lemma
follows due to the equality $\density{\difftables}[i,j]=\bar{\difftables}[i,j] - \bar{\difftables}[i-1,j]$.
\end{proof}

Lemma~\ref{lemma:left-inc-b-main} gives the following corollary.
\begin{corollary}\label{corollary:left-b-min-max-pivotal}
$(i,j)$ is a step index in $\difftables$ if and only if $j$ is equal to the
maximum column index among all pivotal points in rows
$i+1,\ldots,\nextmatch{i}{\sigma}{\bstring}$ of the matrix $\density{\ssam}$.
\end{corollary}

\subsection{$1$-sided incremental $\ssam$ matrix}
\label{subsection:increment-ssam}
We now describe our algorithm that supports $1$-sided operations to the string
$\astring$. We consider the incremental operation of prepending a character
$\sigma$ to $\astring$ (an appending operation to $\astring$
can be carried out in the same manner).
In our proposed solution, we assume that there is an auxilary
data structure that encodes the string $\bstring$ and allows access operations
to any position in $\bstring$ in constant time.
We do not take into account the space occupied by this data structure 
in the space complexity of the algorithm. This leads to an $O(\Delta)$
space data structure that supports prepend (or append) operations to the
string $\astring$ in $O(\Delta)$ time.

Recall that the all scores matrix $\ssam$ is encoded via the
pivotal points of its density matrix. These points are stored in a list $P$, sorted
by increasing row indices.
Our algorithm computes the step indices of $\difftables$, from which
the pivotal points of $\difftables^\square$, and hence the pivotal points of $\ssam'^\square$,
can be obtained (Lemma~\ref{lemma:pivotal-points}).
Note that the number of step indices in $\difftables$ is bounded by $\Delta$.
This is due to the fact that each such step index corresponds to a unique
pivotal point of $\ssam^\square$ (Lemma~\ref{lemma:left-inc-b-main}).

\begin{figure}
\begin{minipage}{0.48\linewidth}

\subfigure[$\density{\ssam}$]{
% \resizebox{\textwidth}{!}{
\input{k-ssam-algorithm-example}
}
\end{minipage}
\hfill
\begin{minipage}{0.48\linewidth}
% \vspace{-4cm}
% \subfigure{
%  \resizebox{1\textwidth}{!}{
% \begin{tabular}{l|llllllllll}
% $i$                               & 0 & 1 & 2 & 3 & 4 & 5 & 6 & 7 & 8 & 9        \\ \hline
% $\nextmatch{i}{\sigma}{\astring}$ & 3 & 3 & 3 & 5 & 5 & 8 & 8 & 8 & 9 & $\infty$
% \end{tabular}
\begin{tabular}{ll}
\multicolumn{1}{l|}{$i$} & $\nextmatch{i}{a}{\bstring}$ \\ \hline
\multicolumn{1}{l|}{0}   & 3                                 \\
\multicolumn{1}{l|}{1}   & 3                                 \\
\multicolumn{1}{l|}{2}   & 3                                 \\
\multicolumn{1}{l|}{3}   & 5                                 \\
\multicolumn{1}{l|}{4}   & 5                                 \\
\multicolumn{1}{l|}{5}   & 8                                 \\
\multicolumn{1}{l|}{6}   & 8                                 \\
\multicolumn{1}{l|}{7}   & 8                                 \\
\multicolumn{1}{l|}{8}   & 9                                 \\
\multicolumn{1}{l|}{9}   & $\infty$            
\end{tabular}

% }
% }
% \centering
% \hspace{5pt}
\subfigure{
 $P = \{ (1,3), (2,5), (3,4), (5,7), (6,8), (8,9) \}$
 }
\end{minipage}
\caption{A run of the algorithm
on the strings $\astring$ and $\bstring$ of Figure~\ref{figure:inc-b-left-ssam},
and a prepend of the character $a$ to $\astring$.
Figure~(a) shows the pivotal points of $\density{\ssam}$.
The red lines correspond to the different values of
$\nextmatch{i}{a}{\bstring}$.
$b_1,b_2$ and $b_3$ denote blocks of consecutive pivotal points as defined in
Section~\ref{subsection:2-sided-ssam}.
The algorithm begins with the pivotal point at row $i_1+1=1$.
It then scans for the next match,
denoted by a red line at index $k_1 = 3$, and verifying at each step that
there is a pivotal point in every row. Once reaching index 3, the algorithm
backward scans the previously traversed pivotal points starting with the
pivotal point at row~$3$.
At each step the algorithm computes the maximum column index among all scanned
pivotal points.
The first column maxima is $j_2=4$, hence $(2,4)$ is a step index in
$\difftables$.
The following pivotal point is $(2,5)$, thus the column maxima is $j_1 = 5$,
and $(1,5)$
is identified as a step index in $\difftables$. The last scanned pivotal point
is $(1,3)$, thus the column maxima remains 5 and the step index identified is
$(0,5)$. The next scanned block is $b_2$ starting with pivotal point $(5,7)$.
The next match point is $5$, and thus $(4,7)$ is a step index. Pivotal point
$(6,8)$ does not yield a step index since there is no match point at index $6$
and also there is no pivotal point at row $7$, hence the search stops after the
first iteration.
The last pivotal point is handled similarly.}
\label{figure:k-ssam-algorithm-example}
\end{figure}

The step indices are computed as follows (see Figure~\ref{figure:k-ssam-algorithm-example}).
Denote by $i_1+1$ the minimum row index of a pivotal point of $\ssam$
(the value of $i_1+1$ is obtained by accessing the first element of $P$) and
let $k_1=\nextmatch{i_1}{\sigma}{\bstring}$.
Note that by Lemma~\ref{lemma:left-inc-b-main}, there is a step index in the
$i_1$'th row of $\difftables$ only if there are $k_1-i_1$ consecutive pivotal
points in $P$ with row indices $i_1+1$ to $k_1$. Therefore we scan
simultaneously the list $P$ and the string $\bstring$ starting from
index $i=i_1+1$.
At each step we check whether the current examined element in $P$ is a pivotal
point in row $i$, and if it is, we move to the next element of $P$ and
increase $i$ by $1$.
This scan is stopped when reaching the index $k_1$ for which $\bstring[k_1] =
\sigma$.
If such $k_1$ is found then by Lemma~\ref{lemma:left-inc-b-main}, every row from
rows $i_1,\ldots,k_1-1$ contains a step index. In order to compute these step
indices, go over $i=k_1-1,k_1-2,\ldots,i_1$.
For each such $i$, compute the  maximum column index $j_i$ among all
pivotal points in rows $i+1,\ldots,k_1$.
$j_i$ can be computed in constant time since $j_i$ is equal to the maximum
of $j_{i+1}$ and the column of the pivotal point of row $i+1$
(this column is obtain from $P$, which is scanned in reverse order,
starting from the element holding the pivotal point of row $k_1$).
By Corollary~\ref{corollary:left-b-min-max-pivotal}, $(i,j_i)$ is the step index of row $i$.

% Then, traverse the list $P$ from the start of the list until reaching
% a pivotal point with row index $i_2+1$ that is larger than $k$.

The algorithm processes the remaining pivotal points (starting from the topmost
pivotal point below row $k_i$) similarly, until
exhausting all $\Delta$ pivotal points.

\paragraph{Complexity analysis}
To obtain the set of step indices, we traverse the list $P$ of pivotal points,
and examine each pivotal point at most twice (once during a forward scan on $P$
and once during a backward scan on $P$).
We also scan the string $\bstring$. The number of access
operation to $\bstring$ is bounded by the number of pivotal points.

Once the set of step indices has been obtained, the
computation of $\difftables^\square$ (using Lemma~\ref{lemma:pivotal-points})
and the pivotal points of $\ssam'$
is done in $O(\Delta)$ time. Hence the total running time is $O(\Delta)$
and the space complexity is $O(\Delta)$. 

\paragraph{Appending a character to $\bstring$}
Appending character $\sigma$ to $\astring$ follows the same paradigm. 
Note that now $\ssam' [i,j] = \max \{
\ssam[i,\prevmatch{j}{\sigma}{\bstring}-1]+1, \ssam[i,j] \}$.

The next lemma summarizes the
properties required to carry out an append operation to $\astring$.
% , it unifies Lemma~\ref{lemma:left-inc-b-main} and
% Corollary~\ref{corollary:left-b-min-max-pivotal}.

\begin{lemma}\label{lemma:right-inc-b-main}
Let $k=\prevmatch{j}{\sigma}{\bstring}$.
$\difftables [i,j]=1$ if and only if there is one pivotal point in each column
from columns $k,\ldots,j$ of the matrix $\density{\ssam}$, and $i$ is less
than the minimum row index among all pivotal points in the submatrix
$\submat{\ssam^\square}{1}{n}{k}{j}$.
\end{lemma}

\begin{proof}
Following the same steps as in Lemma~\ref{lemma:left-inc-b-main} we obtain that
$\difftables [i,j]=1$ if and only if:
\begin{equation}\label{equation:right-inc-b}
\sum\limits_{\substack{1 \leq i' \leq i 
\\ k \leq j' \leq j}}\ssam^{\square}[i',j'] = \ssam[0,k-1] - \ssam[0,j].
\end{equation}
Note that $\sum\limits_{\substack{1 \leq i' \leq i  \\ k \leq j' \leq
j}}\ssam^{\square}[i',j'] \geq 0$ and  $\ssam[0,k-1] \leq \ssam[0,j]$.
Hence, Equation~(\ref{equation:right-inc-b}) holds if and only if, (\rn{1})
$\ssam[0,k-1] = \ssam[0,j]$, and (\rn{2}) $\sum\limits_{\substack{1 \leq i' \leq
i \\ k \leq j' \leq j}}\ssam^{\square}[i',j'] = 0$.
The first item holds if and only if $\sum\limits_{\substack{1 \leq i' \leq n 
\\ k \leq j' \leq j}}\ssam^{\square}[i',j'] = j-k+1$, that is, there is a
pivotal point in every column from columns $k,\ldots,j$. The second item holds if and only if
$i$ is less
than the minimum row index among all pivotal points in the submatrix
$\submat{\ssam^\square}{1}{n}{k}{j}$.
\end{proof}

Due to Lemma~\ref{lemma:right-inc-b-main} we need to use a different definition
of step indices in this case.
We say that $(i,j)$ is a \emph{column step index} if
$\difftables[i,j] \neq \difftables[i-1,j]$.
%$\difftables[i,j]$ contains the first zero element at column $j$ of $\difftables$.
Hence, the main difference to the algorithm is that now we need to store 
a list $P'$ containing the pivotal points of $\ssam^\square$ 
sorted by decreasing columns indices.

We obtain the following theorem.

\begin{theorem}\label{theorem:1-sided-ssam}
Given strings $\astring$ and $\bstring$ of lengths $m$ and $n$,
respectively, over alphabet $\Sigma$ with $L=LCS(\astring,
\bstring)$ and $\Delta=n-L$.
We can construct an $O(\Delta)$ space data structure that encodes the LCS score
between $\astring$ and every substring of $\bstring$ and supports 
$1$-sided incremental operations to $\astring$ in $O(\Delta)$ time.
\end{theorem}

\subsection{$2$-sided incremental $\ssam$ matrix}\label{subsection:2-sided-ssam}
In this section we show how to extend the previously discussed data structure to
support both prepend and append operations to $\astring$. This extension
requires $O(n)$ space and $O(\Delta)$ time per operation.

In order to support incremental operations for both sides of the string
$\astring$, one possible solution is to sort the elements of $P$ by their
columns indices to obtain a list $P'$ when an append operation is performed.
This sorting takes $O(\Delta\log\log\Delta)$ time using the sorting algorithm
of~\cite{han2004deterministic}.
%However sorting may take $O(\Delta \log \Delta)$ using classical sorting
%algorithm (e.g. MergeSort), or $O(\Delta \log \log m)$ by exploiting the fact
%that the indices are bounded integers in the range $1\ldots m$ and applying Van
%Emde Boas priority queue~\cite{van1976design}.
To obtain $O(\Delta)$ worst-case running time we use the observation that we
don't need the lists $P$ and $P'$ to be sorted, and instead we will use a weaker
requirement on the order of the elements.

We say that $[j_1:j_2]$ is a \emph{column block} if there is a pivotal point
in column $j$ for every $j\in[j_1:j_2]$, and there are no pivotal points in
column $j_1-1$ and column $j_2+1$.
We now store $P'$ in the following order.
For every column block $[j_1:j_2]$, the pivotal points in columns
$[j_1:j_2]$ appear consecutively in $P'$, and ordered by decreasing
column indices (namely, the pivotal point in
column $j_2$ appears first, then the pivotal point in column $j_2-1$ an so on).
This corresponds to the sorting required for the appending operation to
$\astring$ as described in Section~\ref{subsection:increment-ssam}.
Note that we do not require a specific order between pivotal points of different
blocks.
We define \emph{row block} analogously and we store $P$ ordered according
to the row blocks.
It is easy to verify that the algorithm described in
Section~\ref{subsection:increment-ssam} remains correct when $P$ and $P'$ are
stored in block order.

%We denote by \emph{$P$-block} (resp. $P'$-block) a set of consecutive pivotal
%points in $P$ (resp. $P'$), and simply say \emph{block} when the respective
%list is known from context. We say that a block is maximal
%if it is not a proper subset of any other block.
%Note that it is sufficient to
%identify blocks for the workflow of the algorithm, as the step indices
%that correspond to each such block can be computed independently. Hence the
%order of the blocks is irrelevant and its computation is excessive.

We now describe how to construct the list $P'$ from $P$.
The process requires an auxiliary array $\arrayName$ of size
$n$, in which each cell is initialized with 0.
%This array is used to find the column blocks.
Let $(i_1,j_1),\ldots,(i_\Delta,j_\Delta)$ be the pivotal points.
We traverse the set of pivotal points as stored in $P$ and set
$\arrayAt{j_t} \leftarrow i_t$
for every $1 \leq t \leq \Delta$.
Note that after the step above,
every column block corresponds to a maximal sub-array
of $\arrayName$ with non-zero elements.
We next traverse again the set of pivotal points in order to extract the column
blocks.
For each pivotal point $(i_t,j_t)$ we examine the value $\arrayAt{j_t}$.
If $\arrayAt{j_t} \neq 0$, we scan the array $\arrayName$ starting at index
$j_t$ to find the minimum index $j' > j_t$ such that $\arrayAt{j'} = 0$.
Then, for $j = j'-1,j'-2,\ldots$ we add the pivotal point
$(\arrayAt{j},j)$ to the end of $P'$, and set $\arrayAt{j}$ to $0$.
This loop is stopped when $\arrayAt{j} = 0$.
We then move to the next pivotal point
$(i_{t+1},j_{t+1})$.
The running time of this algorithm is $O(\Delta)$.

%For each block we apply the same approach described in
%Section~\ref{subsection:increment-ssam} to obtain the encoding of $\ssam'$. 
We obtain the following theorem.

\begin{theorem}\label{theorem:2-sided-ssam}
Given strings $\astring$ and $\bstring$ of lengths $m$ and $n$,
respectively, over alphabet $\Sigma$ with $L=LCS(\astring,
\bstring)$ and $\Delta=n-L$.
We can construct an $O(n)$ space data structure that encodes the LCS score
between $\astring$ and every substring of $\bstring$ and supports $2$-sided
incremental operations to both ends of $\astring$ in $O(\Delta)$ time.
\end{theorem}

We note that our approach cannot support prepend or append operations on
$\bstring$. The reason is that such operations increase the size of the matrix
$\ssam$. The new column in $\density{\ssam}$ may contain a new pivotal point.
However, the matrix $\density{\ssam}$ does not have the information required for computing
this new pivotal point.

\section{Incremental $\psam$ matrix}\label{section:k-psam}
Our method introduced above can be modified to maintain the
matrix $\psam$ and supports prepend
operations to $\astring$ and append operations to $\bstring$.
Note that the size of the matrix $\psam$ is $n \times m$, thus it
increases when such operations are performed.
We start by considering an incremental operation that prepends a character
$\sigma$ to $\astring$, that is $\astring' = \sigma \astring$. We assume
that $\psam$ is over the index set $[0:n] \times [0:m]$ and $\psam'$  is over
the index set $[0:n] \times [-1:m]$.
We define $\difftables [i,j] = \psam' [i,j] - \psam[i,j]$ for $j \geq 0$,
and $\difftables[i,j]=0$ otherwise.

In this case we have that $\psam'[i,j] = \max \{ \psam [i,j], \psam
[\nextmatch{i}{\sigma}{\bstring},j]+1 \}$ for $j \geq 0$.
% Consequently, Lemma~\ref{lemma:left-inc-b-main} is modified accordingly.
The following lemma is the analogous of Lemma~\ref{lemma:left-inc-b-main}.

\begin{lemma}\label{lemma:left-inc-a-main}
For $j\geq 0$, $\difftables [i,j]=1$ if and only if $\nextmatch{i}{\sigma}{\bstring} <
\infty$ and there are no pivotal points in the submatrix
$\submat{\psam^\square}{i+1}{\nextmatch{i}{\sigma}{\bstring}}{1}{j}$.
\end{lemma}

\begin{proof}
Let $\nextmatch{i}{\sigma}{\bstring}=k$.
Similarly to
Lemma~\ref{lemma:left-inc-b-main}, $\difftables [i,j] =1$ if and only if
$\sum\limits_{\substack{i+1 \leq i' \leq k  \\ 1
\leq j' \leq j}}\psam^{\square}[i',j'] = \psam[i,0] - \psam[k,0]$.
The lemma follows due to the definition of $\psam$, since $\psam
[i,0]=\psam [k,0]=0$.
\end{proof}

This leads to the following corollary (analogous to
Corollary~\ref{corollary:left-inc-b-monotonicity-diff-table}).

\begin{corollary}\label{corollary:left-inc-a-monotonicity-diff-table}
If $\difftables [i,j] = 1$ then $\difftables [i,j'] = 1$ for every $0 \leq j' \leq j$.
\end{corollary}

Following Corollary~\ref{corollary:left-inc-a-monotonicity-diff-table} we say
that $(i,j)$ is a \emph{step index} if $j \geq 1$ and
$\difftables[i,j] \neq \difftables[i,j-1]$
%$\difftables [i,j] - \difftables[i,j-1]=-1$
(see Figure~\ref{figure:left-inc-a}).
We then obtain the following corollary (analogous to
Corollary~\ref{corollary:left-b-min-max-pivotal}).

\begin{corollary}\label{corollary:left-a-min-max-pivotal}
$(i,j)$ is a step index in $\difftables$ if and only if $\nextmatch{i}{\sigma}{\bstring} <
\infty$ and $j$ is equal to the
minimum column index among all pivotal points in rows
$i+1,\ldots,\nextmatch{i}{\sigma}{\bstring}$ of the matrix $\density{\psam}$.
\end{corollary}

Lemma~\ref{lemma:left-inc-a-main} introduces several challenges compared to
Section~\ref{section:k-ssam}. First, the number of step indices is not bounded
by the number of pivotal points (see for example
Figure~\ref{figure:left-inc-a}). Consequently, the number of step indices may be
$\Omega(n)$. Thus, to obtain $O(L)$ running time, we show in
Lemma~\ref{lemma:left-inc-a-density-psam} how to
compute $\density{\difftables}$ directly, without computing the entire set of
step indices explicitly. Furthermore, the property of consecutive pivotal points,
which was exploited in Section~\ref{section:k-ssam} to bound the time
complexity of scanning for the next character matching to $\sigma$,
no longer holds.
Thus, we need to use a data structure that supports
$\nextmatch{i}{\sigma}{\bstring}$ or $\prevmatch{j}{\sigma}{\bstring}$ queries in
constant time.

%Since $\difftables [i,-1] = 0$ for every $0 \leq i \leq n$, we have that if
%$(i,j)$ is a step index then $j$ is greater than 0.
Following Lemma~\ref{lemma:left-inc-a-main} we have that
$\difftables[i,0] = 1$ for every $0 \leq i \leq n$ such that
$\nextmatch{i}{\sigma}{\bstring} < \infty$. Hence, if $\difftables [i,0] = 1$
then $\difftables [i',0] = 1$ for every $0 \leq i' \leq i$.
We get that the maximal row index $i$ for which $\difftables [i,0] = 1$ yields
a pivotal point in cell $\density{\difftables}[i+1,0]$ of value -1 (see for
example $\density{\difftables}[7,0]$ in Figure~\ref{figure:left-inc-a}~(d)).
This yields the following corollary.

\begin{corollary}\label{corollary:left-inc-a-special-pivotal}
$\density{\difftables}[i,0]=-1$ for the minimum index $i$ for which
$\nextmatch{i}{\sigma}{\bstring} = \infty$, and this is the only pivotal point
in the 0'th column of $\density{\difftables}$.
%  $i = \min \{ 0\leq i'\leq n :
% \nextmatch{i'}{\sigma}{\bstring} = \infty \}$.
\end{corollary}

From the definition of step indices we obtain the following lemma.

\begin{lemma}\label{lemma:left-inc-a-pivotal-points}
For every cell $(i,j)$ with $j>0$ in $\density{\difftables}$,
\begin{itemize}
  \item $\density{\difftables}[i,j]=-1$ if and only if $(i,j)$ is a step index
  in $\difftables$ and $(i-1,j)$ is not a step index.
  \item $\density{\difftables}[i,j]=1$ if and only if $(i,j)$ is not a step
  index in $\difftables$ and $(i-1,j)$ is a step index.
\end{itemize}
\end{lemma}

We would like to identify cases (\rn{1}) and (\rn{2}) described above,
to be able to compute the pivotal points of $\density{\difftables}$ directly.

\begin{lemma}\label{lemma:left-inc-a-density-psam}
Let $j_{\min}$ be the minimum column index among the pivotal points
in rows $i+1,\ldots,\nextmatch{i}{\sigma}{\bstring}$. If there are no
such pivotal points, $j_{\min} = \infty$. In what follows we consider only cells
of $\density{\difftables}$ with column indices greater than 0.
\begin{itemize}
\item
If $\bstring[i] \neq \sigma$ then row $i$ of $\density{\difftables}$
contains non zero elements if and only if
$\nextmatch{i}{\sigma}{\bstring} < \infty$ and there is a pivotal point $(i,j)$
in $\density{\psam}$ with $j < j_{\min}$.
If these conditions hold, $\density{\difftables}[i,j]=1$.
Additionally, if $j_{\min} < \infty$, $\density{\difftables}[i,j_{\min}]=-1$.
All other cells at row $i$ of
$\density{\difftables}$ contain zeros.
\item
If $\bstring[i] = \sigma$ then $\density{\difftables}[i,j_{\min}]=-1$ if
$\nextmatch{i}{\sigma}{\bstring} < \infty$ and $j_{\min} < \infty$.
Additionally, if there is a pivotal point $(i,j)$ in $\density{\psam}$
then $\density{\difftables}[i,j]=1$.
All other cells at row $i$ of
$\density{\difftables}$ contain zeros.
\end{itemize}
\end{lemma}

\begin{proof}
\begin{comment} ALTERNATIVE PROOF
For the first part of the corollary, if $\bstring [i] \neq \sigma$ then
$\nextmatch{i-1}{\sigma}{\bstring} = \nextmatch{i}{\sigma}{\bstring}$. Row $i$
of $\density{\difftables}$ contains non zero element if and only if there
exist an index $j'$ for which the first or the second case of
Lemma~\ref{lemma:left-inc-a-pivotal-points} holds. If the first case holds, that
is $(i,j')$ is a step index in $\difftables$ and $(i-1,j')$ is not a step index,
then by Corollary~\ref{corollary:left-a-min-max-pivotal} we get that
$\nextmatch{i}{\sigma}{\bstring} < \infty$, $j' = j_{\min} < \infty$ and there
is a pivotal point $(i,j)$ in $\density{\psam}$ with $j<j_{\min}$. Moreover, $\density{\difftables}[i,j_{\min}]=-1$. 
If the second case holds, that is $(i-1,j)$ is a step index in
$\difftables$ and $(i,j)$ is not a step index, then $\nextmatch{i}{\sigma}{\bstring} < \infty$, there
is a pivotal point $(i,j)$ in $\density{\psam}$ with $j<j_{\min}$ and $j'=j$.
Moreover, $\density{\difftables}[i,j]=1$.
By Lemma~\ref{lemma:left-inc-a-pivotal-points} these are the only non zero cells
in row $i$ of $\density{\difftables}$.
\end{comment}
Row $i$
of $\density{\difftables}$ contains non zero element if and only if there
is a cell in this row for which the first or the second case of
Lemma~\ref{lemma:left-inc-a-pivotal-points} holds.

For the first part of the lemma, since $\bstring [i] \neq
\sigma$, we get $\nextmatch{i-1}{\sigma}{\bstring} =
\nextmatch{i}{\sigma}{\bstring}$.
If $\nextmatch{i}{\sigma}{\bstring} = \infty$ then
also $\nextmatch{i-1}{\sigma}{\bstring} = \infty$ and
by Corollary~\ref{corollary:left-a-min-max-pivotal} rows $i-1$ and $i$
do not have step indices.
Therefore, by Lemma~\ref{lemma:left-inc-a-pivotal-points}
row $i$ does not contain non zero elements.
Assume now that $\nextmatch{i}{\sigma}{\bstring} < \infty$.
If there is no pivotal point $(i,j)$ with $j < j_{\min}$ then by
Corollary~\ref{corollary:left-a-min-max-pivotal},
both rows $i-1$ and $i$ have step indices at column $j_{\min}$ 
if $j_{\min} < \infty$, and these rows do not have step indices if
$j_{\min} = \infty$.
Hence, neither case of Lemma~\ref{lemma:left-inc-a-pivotal-points} can occur,
so row $i$ does not contain non zero elements.

Suppose now that there is a pivotal point $(i,j)$ with $j < j_{\min}$.
By Corollary~\ref{corollary:left-a-min-max-pivotal},
$(i,j)$ is not a step index in $\difftables$ and $(i-1,j)$ is a step index.
By Lemma~\ref{lemma:left-inc-a-pivotal-points}, $\density{\difftables}[i,j]=1$.
Moreover, if $j_{\min} < \infty$ then
$(i,j_{\min})$ is a step index in $\difftables$ and $(i-1,j_{\min})$ is not 
a step index, so $\density{\difftables}[i,j_{\min}]=-1$.
Finally, if $j_{\min} = \infty$ then row $i$ of $\difftables$ does not contain
a step index and the first case of Lemma~\ref{lemma:left-inc-a-pivotal-points}
cannot occur.

For the second part of the lemma, since $\bstring [i] =
\sigma$, we get $\nextmatch{i-1}{\sigma}{\bstring} = i$.
Suppose that $\nextmatch{i}{\sigma}{\bstring} < \infty$ and $j_{\min} < \infty$.
By Corollary~\ref{corollary:left-a-min-max-pivotal},
$(i,j_{\min})$ is a step index in $\difftables$ and $(i-1,j_{\min})$ is not
a step index (since the rows range $(i-1)+1,\ldots,\nextmatch{i-1}{\sigma}{\bstring}$ consists of only row $i$,
and row $i$ cannot have a pivotal point in column $j_{\min})$.
By Lemma~\ref{lemma:left-inc-a-pivotal-points},
$\density{\difftables}[i,j_{\min}]=-1$.
Additionally, if there is a pivotal point $(i,j)$ in $\density{\psam}$
then $(i,j)$ is not a step index in $\difftables$ and $(i-1,j)$ is a
a step index. Therefore, $\density{\difftables}[i,j]=1$.
\end{proof}

% 
% The following proposition is obtained similarily to
% Proposition~\ref{proposition:left-inc-b}
% \begin{proposition}\label{proposition:left-inc-a}
% $\psam'[i,j] = \ssam[i,j] + 1 $ if and only if $\nextmatch{i}{\sigma} < \infty$
% and $\psam [i,j]=\ssam [\nextmatch{i}{\sigma},j]$.
% \end{proposition}

\begin{figure}[]
\label{figure:left-inc-a}
\centering
\begin{minipage}{0.4\linewidth}

\subfigure[$\psam'$]{
\resizebox{1.3\textwidth}{!}{%
\begin{tabular}{|l|l|l|l|l|l|l|l|l|}
\hline
0 & 1                         & 2                         & 3                         & 3 & 3 & 4                         & 5                         & 5 \\ \hline
0 & 1                         & 2                         & 3                         & 3 & 3 & 4                         & 5                         & 5 \\ \hline
0 & 1                         & 2                         & 3                         & 3 & 3 & \cellcolor[HTML]{C0C0C0}3 & 4                         & 4 \\ \hline
0 & 1                         & 2                         & 3                         & 3 & 3 & 3                         & 4                         & 4 \\ \hline
0 & 1                         & \cellcolor[HTML]{C0C0C0}1 & 2                         & 2 & 2 & 2                         & 3                         & 3 \\ \hline
0 & 1                         & 1                         & \cellcolor[HTML]{C0C0C0}1 & 1 & 1 & 1                         & 2                         & 2 \\ \hline
0 & 1                         & 1                         & 1                         & 1 & 1 & 1                         & \cellcolor[HTML]{C0C0C0}1 & 1 \\ \hline
0 & \cellcolor[HTML]{C0C0C0}0 & 0                         & 0                         & 0 & 0 & 0                         & 0                         & 0 \\ \hline
\end{tabular}
}
}

\subfigure[$\psam$]{
% \hspace{15pt}
\resizebox{1.3\textwidth}{!}{%
\begin{tabular}{|l|l|l|l|l|l|l|l|}
\hline
0 & 1                         & 2                         & 2 & 2 & 3                         & 4                         & 4 \\ \hline
0 & 1                         & 2                         & 2 & 2 & 3                         & 4                         & 4 \\ \hline
0 & 1                         & 2                         & 2 & 2 & 3                         & 4                         & 4 \\ \hline
0 & 1                         & 2                         & 2 & 2 & \cellcolor[HTML]{C0C0C0}2 & 3                         & 3 \\ \hline
0 & 1                         & 2                         & 2 & 2 & 2                         & 3                         & 3 \\ \hline
0 & 1                         & \cellcolor[HTML]{C0C0C0}1 & 1 & 1 & 1                         & 2                         & 2 \\ \hline
0 & \cellcolor[HTML]{C0C0C0}0 & 0                         & 0 & 0 & 0                         & 1                         & 1 \\ \hline
0 & 0                         & 0                         & 0 & 0 & 0                         & \cellcolor[HTML]{C0C0C0}0 & 0 \\ \hline
\end{tabular}
}
}

\end{minipage}
\hfill
\begin{minipage}{0.4\linewidth}
% \vspace{-10pt}
\subfigure[$\difftables$]{
\resizebox{1.2\textwidth}{!}{%
\begin{tabular}{|l|l|l|l|l|l|l|l|l|}
\hline
0 & 1 & 1 & 1 & 1 & 1 & 1 & 1 & 1 \\ \hline
0 & 1 & 1 & 1 & 1 & 1 & 1 & 1 & 1 \\ \hline
0 & 1 & 1 & 1 & 1 & 1 & \cellcolor[HTML]{E8E8E8}0 & 0 & 0 \\ \hline
0 & 1 & 1 & 1 & 1 & 1 & 1 & 1 & 1 \\ \hline
0 & 1 & \cellcolor[HTML]{E8E8E8}0 & 0 & 0 & 0 & 0 & 0 & 0 \\ \hline
0 & 1 & \cellcolor[HTML]{E8E8E8}0 & 0 & 0 & 0 & 0 & 0 & 0 \\ \hline
0 & 1 & 1 & 1 & 1 & 1 & 1 & \cellcolor[HTML]{E8E8E8}0 & 0 \\ \hline
0 & 0 & 0 & 0 & 0 & 0 & 0 & 0 & 0 \\ \hline
\end{tabular}
}
}

\subfigure[$\difftables^\square$]{
\resizebox{1.2\textwidth}{!}{%
\begin{tabular}{|l|l|l|l|l|l|l|l|}
\hline
0                          & 0                          & 0 & 0 & 0 & 0                          & 0                          & 0 \\ \hline
0                          & 0                          & 0 & 0 & 0 & \cellcolor[HTML]{C0C0C0}-1 & 0                          & 0 \\ \hline
0                          & 0                          & 0 & 0 & 0 & \cellcolor[HTML]{C0C0C0}1  & 0                          & 0 \\ \hline
0                          & \cellcolor[HTML]{C0C0C0}-1 & 0 & 0 & 0 & 0                          & 0                          & 0 \\ \hline
0                          & 0                          & 0 & 0 & 0 & 0                          & 0                          & 0 \\ \hline
0                          & \cellcolor[HTML]{C0C0C0}1  & 0 & 0 & 0 & 0                          & \cellcolor[HTML]{C0C0C0}-1 & 0 \\ \hline
\cellcolor[HTML]{C0C0C0}-1 & 0                          & 0 & 0 & 0 & 0                          & \cellcolor[HTML]{C0C0C0}1  & 0 \\ \hline
\end{tabular}
}
}
\end{minipage}
\caption{An example of the matrices $\psam', \psam,\difftables$ and
$\difftables^\square$ for $\astring = bbcbbaa$, $\bstring = aacabba$,
and a predend of the character $a$ to $\astring$.
Pivotal points are colored dark gray and step indices in $\difftables$ are
colored light gray.}

\end{figure}

\subsection{$1$-sided incremental $\psam$ matrix}
\label{subsection:increment-psam}
In this section we encode the matrix $\psam$ such that prepend operations to
the string $\astring$ are supported. We encode the matrix $\psam$ by the $L$ pivotal
points of its density matrix. These points are stored in a list $P$ sorted
by increasing row indices. Note that the string $\bstring$ remains
constant, hence we can preprocess $\bstring$ in $O(n)$ time and
compute two static look-up tables that contain all possible
values of $\nextmatch{i}{\sigma}{\bstring}$ and $\prevmatch{i}{\sigma}{\bstring}$. 
We then use
Lemma~\ref{lemma:left-inc-a-density-psam} to compute
$\density{\difftables}$ efficiently in $O(L)$ time, see
Figure~\ref{figure:k-psam-algorithm-example} for a running example.

We begin iterating over $P$.
Let $(i_1+1,j_1)$ denote the first element of $P$.
We compute
$k_1^- = \prevmatch{i_1}{\sigma}{\bstring}$, and $k_1^+ =
\nextmatch{i_1}{\sigma}{\bstring}$. 
$k_1^-$ and $k_1^+$ are the starting and ending indices of a block of pivotal
points (in contrast to Section~\ref{subsection:2-sided-ssam},
here the pivotal points of a block do not necessarily have
consecutive row indices).
We scan the list $P$ until reaching a pivotal point $(i_2+1,j_2)$ for which
$i_2+1> k_1^+$ (forward scan).
We then scan backward the block
starting from the pivotal point preceding $(i_2+1,j_2)$, and at each step we compute the minimum column
index in the block so far and apply
Lemma~\ref{lemma:left-inc-a-density-psam}
to obtain the pivotal points of
$\density{\difftables}$.
% \textcolor{blue}{for dekel. if the minimum is replaced
%from $j'$ to $j''$ ($j''<j'$) at pivotal $(i'',j'')$ then
%$\density{\difftables}[i'',j'']=1$ and $\density{\difftables}[i'',j']=-1$ via
%Lemma~\ref{lemma:left-inc-a-density-psam}}.

The algorithm processes the remaining pivotal points (starting from the
pivotal point $(i_2+1,j_2)$) similarly, until
exhausting all $L$ pivotal points. At the end we compute
$\prevmatch{n}{\sigma}{\bstring}$ to obtain the pivotal point at column 0 by
Corollary~\ref{corollary:left-inc-a-special-pivotal}.

\paragraph{Complexity analysis} The preprocessing step takes $O(n)$ time and
space. At each step of the algorithm we examine each pivotal point at most
twice, hence we obtain $O(L)$ time and $O(n)$ space.

The following theorem concludes the data structure described in this section. We
note that the append operation to $\bstring$ can be carried out similarly.

\begin{theorem}\label{theorem-1-sided-psam}
Given strings $\astring$ and $\bstring$ of lengths $m$ and $n$,
respectively, over alphabet $\Sigma$ with $L=LCS(\astring,
\bstring)$.
We can construct an $O(n)$ (resp. $O(m)$) space data structure that encodes the
LCS score between every suffix of $\bstring$ and every prefix of $\astring$ and
supports $1$-sided prepend operations to $\astring$ (resp. append operations
to $\bstring$) in $O(L)$ time.
\end{theorem}

\begin{figure}
\label{figure:k-psam-algorithm-example}

\begin{minipage}{0.4\linewidth}

\subfigure[$\density{\psam}$]{
% \resizebox{\textwidth}{!}{
\input{k-psam-algorithm-example}
}
\end{minipage}
\hfill
\begin{minipage}{0.55\linewidth}
% \vspace{-4cm}
% \subfigure{
%  \resizebox{1\textwidth}{!}{
% \begin{tabular}{l|llllllllll}
% $i$                               & 0 & 1 & 2 & 3 & 4 & 5 & 6 & 7 & 8 & 9        \\ \hline
% $\nextmatch{i}{\sigma}{\astring}$ & 3 & 3 & 3 & 5 & 5 & 8 & 8 & 8 & 9 & $\infty$
% \end{tabular}
\begin{tabular}{l|l|l}
$i$ & $\prevmatch{i}{a}{\bstring}$ & $\nextmatch{i}{a}{\bstring}$ \\ \hline
0   & $-\infty$                         & 1                                 \\
1   & 1                                 & 2                                 \\
2   & 2                                 & 4                                 \\
3   & 2                                 & 4                                 \\
4   & 4                                 & 7                                 \\
5   & 4                                 & 7                                 \\
6   & 4                                 & 7                                 \\
7   & 7                                 & $\infty$                          
\end{tabular}

% }
% }
% \centering
% \hspace{5pt}
\subfigure{
 $P = \{ (3,5), (5,2), (6,1), (7,6) \}$
 }
\end{minipage}
\caption{A run of the algorithm
on the strings $\astring$ and $\bstring$ of Figure~\ref{figure:left-inc-a},
and a prepend of the character $a$ to $\astring$.
The algorithm begins with $i_1+1=3$, where $k_1^- =
\prevmatch{i_1}{\sigma}{\bstring}=2$, and $k_1^+ =
\nextmatch{i_1}{\sigma}{\bstring}=4$.
The only pivotal point in rows $2,3,4$ is $(3,5)$,
hence  $\density{\difftables}[3,5]=1$ by part~1 of
Lemma~\ref{lemma:left-inc-a-density-psam}
and $\density{\difftables}[2,5]=-1$ by part~2 of
Lemma~\ref{lemma:left-inc-a-density-psam}.
The following scanned pivotal point is $(5,2)$ for which
$k_2^-=4$ and $k_2^+=7$. The backward scan begins with the pivotal point
$(7,6)$ and $j_{\min} = \infty$.
By Lemma~\ref{lemma:left-inc-a-density-psam} set
$\density{\difftables}[7,6]=1$.
The value of $j_{\min}$ is $6$ when scanning the
pivotal point $(6,1)$, thus we set $\density{\difftables}[6,1]=1$ and
$\density{\difftables}[6,6]=-1$.
When scanning the pivotal point $(5,2)$ the value of $j_{\min}$ is 1,
thus row 5 does not contain non zero elements.
The block scan is complete after setting $\density{\difftables}[4,1]=-1$.
Lastly, we set $\density{\difftables}[7,0]=1$ by
Corollary~\ref{corollary:left-inc-a-special-pivotal}.}
\end{figure}

\subsection{$2$-sided incremental $\psam$ matrix}
\label{subsection:2-sided-psam}
For the 2-sided case we note that both $\astring$ and $\bstring$ are subjected
to modifications. Thus, the look-up tables must be updated dynamically.
We suggest the following data structure that supports
$\prevmatchb{\bstring}$ and $\nextmatchb{\bstring}$ queries in constant time,
and append operations to $\bstring$ in constant time.

We start with a data structure for $\prevmatch{i}{\sigma}{\bstring}$. 
Note that the values of $\prevmatch{i}{\sigma}{\bstring}$ remain constant
when append operations to $\bstring$ are applied.
Thus, we store a dynamic array $\prevtableb{\bstring}$ in which
$\prevtable{i}{\sigma}{\bstring} = \prevmatch{i}{\sigma}{\bstring}$
for all $i$ and $\sigma$
(a dynamic array is a data structure that stores an array and allows
appending a constant number of cells to the end of the array in constant
worst-case time).
Assuming that $\bstring$ is of length $n$ and a character $\tau$ is appended to
$\bstring$, we only need to update the cell
$\prevtable{n+1}{\sigma}{\bstring}$ for every $\sigma \in \Sigma$.
Clearly,
$\prevtable{n+1}{\sigma}{\bstring} = \prevtable{n}{\sigma}{\bstring}$ if
$\sigma \neq \tau$ and $\prevtable{n+1}{\tau}{\bstring} = n+1$.

For the data structure that supports $\nextmatch{i}{\sigma}{\bstring}$
queries, it is no longer the case that the values remain constant.
To handle this, we use a dynamic array $\nexttableb{\bstring}$ that is
defined as follows:
$\nexttable{i}{\sigma}{\bstring} = \nextmatch{i}{\sigma}{\bstring}$
if $i = 0$ or $\bstring[i] = \sigma$.
Otherwise, $\nexttable{i}{\sigma}{\bstring}$ contains an arbitrary
value.
We can now compute $\nextmatch{i}{\sigma}{\bstring}$ using
the following equality:
$\nextmatch{i}{\sigma}{\bstring} = \nexttable{i'}{\sigma}{\bstring}$
where $i'=\max(0,\prevmatch{i}{\sigma}{\bstring})$.
Assuming that $\bstring$ is of length $n$ and a character $\tau$ is appended to
$\bstring$, we update $\nexttableb{\bstring}$ by
%setting $\nexttable{n+1}{\sigma}{\bstring} = \infty$ for all $\sigma$.
setting $\nexttable{i}{\tau}{\bstring} = n+1$ where
$i = \prevmatch{n}{\tau}{\bstring}$.

Now, using the dynamic data structures, the technique remains the same as
described in Section~\ref{subsection:increment-psam}. Note that when an append
operation is preformed, the set of pivotal points needs to be sorted by the
column indices. This requires $O(L \log \log L)$ time for sorting the set of $L$
pivotal points using the sorting algorithm
of~\cite{han2004deterministic}.
This leads to the following theorem.

\begin{theorem}\label{theorem:2-sided-psam}
Given strings $\astring$ and $\bstring$ of lengths $m$ and $n$,
respectively, over alphabet $\Sigma$ with $L=LCS(\astring,
\bstring)$.
We can construct an $O(m+n)$ space data structure that encodes the LCS score
between every suffix of $\bstring$ and every prefix of $\astring$ and
supports prepend operations to $\astring$ and append operations to $\bstring$ in
$O(L \log \log L)$ time.
\end{theorem}

\section{Incremental $\ssam$ and $\psam$ matrices}\label{section:ssam+psam}
In this section we show how to maintain the matrices $\ssam$ and $\psam$ and
support append operations to either $\astring$ or $\bstring$.

Recall that $\psam$ is over the index set $[0:n] \times [0:m]$ and $\ssam$ is
over the index set $[0:n] \times [0:n]$. We store a list $P_{\ssam}$ for the
matrix $\ssam$ as defined in Section~\ref{section:k-ssam}, and a list $P_{\psam}$ for the matrix $\psam$ as defined in Section~\ref{section:k-psam}.
Both lists are sorted by decreasing column indices.

First consider an append operation to $\bstring$. This operation
adds a new row to the matrix $\psam$, and a new row and column to $\ssam$. 
Denote the modified matrices by $\psam'$ and $\ssam'$.
The modifications to the matrix $\psam$ are carried as described in
Section~\ref{subsection:increment-psam}. 
The matrix $\ssam'$, on the other hand, does not vary much from the matrix
$\ssam$, since the submatrix
$\submat{\ssam'}{0}{n}{0}{n}$ is not affected by the newly added symbol.
However, note that a new pivotal point
may be added to the matrix $\density{\ssam'}$ at the newly appended column. In
the following lemma we show how to compute this pivotal point.

\begin{lemma}
If $(i,m)$ is a step index in $\difftables = \psam' - \psam$ then $(i,n+1)$ is a
pivotal point in $\density{\ssam'}$. Otherwise, if
$\prevmatch{m}{\sigma}{\astring} = -\infty$ (i.e. there is no new match-point)
then $(n+1,n+1)$ is a pivotal point in $\density{\ssam'}$.
\end{lemma}

\begin{proof}
Note that the $(n+1)$'th column of $\ssam'$ is precisely the $m$'th column of
the matrix $\psam'$, and the $n$'th column of $\ssam$ is precisely the $m$'th
column of $\psam$. Hence if $(i,m)$ is a step index in $\difftables = \psam' -
\psam$ then $\density{\ssam'} [i,n+1] = (\psam'[i,m]-\psam[i,m]) -
(\psam'[i-1,m]-\psam[i-1,m]) = \difftables[i,m] - \difftables[i-1,m] = 1$.
Otherwise, if $\prevmatch{m}{\sigma}{\astring} = -\infty$, then $\ssam' [n,n+1]
= 0$ and by definition we get $\density{\ssam'}[n+1,n+1] = 1$.
\end{proof}

The step index on the $m$'th column of $\psam$ can be computed using
Corollary~\ref{corollary:left-a-min-max-pivotal} in $O(L)$ time.
Hence the modifications to both $\ssam$ and $\psam$ can be carried out in $O(L)$
time.

An append operation to $\astring$ can be carried out analogously.
In this case the matrix $\ssam$ is modified, and the list $P_{\ssam'}$ can be
obtained as described in Section~\ref{subsection:increment-ssam}. A new column
is also appended to $\psam$. In this case we have the following lemma.

\begin{lemma}\label{lemma:schmist-right-inc}
If $(i,n)$ is a step index in $\difftables = \ssam' - \ssam$ then $(i,m+1)$ is a
pivotal point in $\density{\psam'}$.
\end{lemma}

By applying Lemma~\ref{lemma:schmist-right-inc} for the matrix
$\density{\psam'}$ and using the approach described in
Section~\ref{subsection:increment-ssam} for the matrix $\density{\ssam'}$, we
can compute the lists $P_{\psam'}$ and $P_{\ssam'}$ in $O(\Delta)$
time.

We derive the following theorem.

\begin{theorem}\label{theorem:ssam+psam}
Given strings $\astring$ and $\bstring$ of lengths $m$ and $n$,
respectively, over alphabet $\Sigma$ with $L=LCS(\astring,
\bstring)$ and $\Delta=n-L$.
We can construct an $O(m+n)$ space data structure that encodes the LCS score
between $\astring$ and every substring of $\bstring$, and between every suffix
of $\bstring$ and every prefix of $\astring$. This data-structure supports
append operations to $\astring$ or $\bstring$ in $O(\Delta)$ and $O(L)$
time, respectively.
\end{theorem}

\bibliographystyle{abbrv}
\bibliography{bibliography}

\end{document}

%% file: figure_k-iams-ssam-new.tex
\begin{tikzpicture}[scale=0.65, transform
shape, dashed_style/.style={->,dashed}, path_style/.style={->,very thick}]

\tikzstyle{greyvertex}=[circle, draw, inner sep=0pt, minimum
size=6pt, fill=black!50]
\tikzstyle{lightgreyvertex}=[circle, draw, inner sep=0pt, minimum
size=6pt, fill=light gray]

    \draw[thick, ->] (-0.5,0.5) -- (-0.5,-1) node[midway,left] {$\bstring$};
    \draw[thick, ->] (-0.5,0.5) --  (1,0.5) node[midway,above] {$\astring$};
\vertex (c0r0) at (0,-0) {};
\node[greyvertex] (c1r0) at (1,-0) {};
\vertex (c2r0) at (2,-0) {};
\vertex (c3r0) at (3,-0) {};
\vertex (c4r0) at (4,-0) {};
\vertex (c5r0) at (5,-0) {};
\vertex (c6r0) at (6,-0) {};
\node[greyvertex] (c7r0) at (7,-0) {};
\vertex (c0r1) at (0,-1) {};
\node[greyvertex] (c1r1) at (1,-1) {};
\vertex (c2r1) at (2,-1) {};
\vertex (c3r1) at (3,-1) {};
\vertex (c4r1) at (4,-1) {};
\vertex (c5r1) at (5,-1) {};
\vertex (c6r1) at (6,-1) {};
\node[greyvertex] (c7r1) at (7,-1) {};
\vertex (c0r2) at (0,-2) {};
\node[greyvertex] (c1r2) at (1,-2) {};
\vertex (c2r2) at (2,-2) {};
\vertex (c3r2) at (3,-2) {};
\vertex (c4r2) at (4,-2) {};
\vertex (c5r2) at (5,-2) {};
\vertex (c6r2) at (6,-2) {};
\node[greyvertex] (c7r2) at (7,-2) {};
\vertex (c0r3) at (0,-3) {};
\node[greyvertex] (c1r3) at (1,-3) {};
\vertex (c2r3) at (2,-3) {};
\vertex (c3r3) at (3,-3) {};
\vertex (c4r3) at (4,-3) {};
\vertex (c5r3) at (5,-3) {};
\vertex (c6r3) at (6,-3) {};
\node[greyvertex] (c7r3) at (7,-3) {};
\vertex (c0r4) at (0,-4) {};
\node[greyvertex] (c1r4) at (1,-4) {};
\vertex (c2r4) at (2,-4) {};
\vertex (c3r4) at (3,-4) {};
\vertex (c4r4) at (4,-4) {};
\vertex (c5r4) at (5,-4) {};
\vertex (c6r4) at (6,-4) {};
\node[greyvertex] (c7r4) at (7,-4) {};
\vertex (c0r5) at (0,-5) {};
\node[greyvertex] (c1r5) at (1,-5) {};
\vertex (c2r5) at (2,-5) {};
\vertex (c3r5) at (3,-5) {};
\vertex (c4r5) at (4,-5) {};
\vertex (c5r5) at (5,-5) {};
\vertex (c6r5) at (6,-5) {};
\node[greyvertex] (c7r5) at (7,-5) {};
\tikzset{EdgeStyle/.style={->}}
%horizontal edges
\draw[dashed_style] (c0r0) -- (c1r0)  node[midway,above] {c};
\draw[->] (c1r0) -- (c2r0)  node[midway,above] {c};
\draw[->] (c2r0) -- (c3r0)  node[midway,above] {b};
\draw[->] (c3r0) -- (c4r0)  node[midway,above] {c};
\draw[->] (c4r0) -- (c5r0)  node[midway,above] {c};
\draw[->] (c5r0) -- (c6r0)  node[midway,above] {a};
\draw[->] (c6r0) -- (c7r0)  node[midway,above] {a};
\draw[dashed_style] (c0r1) edge node{} (c1r1);
%\Edge (c0r1)(c1r1);
\Edge (c1r1)(c2r1);
\Edge (c2r1)(c3r1);
\Edge (c3r1)(c4r1);
\Edge (c4r1)(c5r1);
\Edge (c5r1)(c6r1);
\Edge (c6r1)(c7r1);
\draw[dashed_style] (c0r2) edge node{} (c1r2);
% \Edge (c0r2)(c1r2);
\Edge (c1r2)(c2r2);
\Edge (c2r2)(c3r2);
\Edge (c3r2)(c4r2);
\Edge (c4r2)(c5r2);
\Edge (c5r2)(c6r2);
\Edge (c6r2)(c7r2);
\draw[dashed_style] (c0r3) edge node{} (c1r3);
% \Edge (c0r3)(c1r3);
\Edge (c1r3)(c2r3);
\draw[path_style] (c2r3) edge node{} (c3r3);
\draw[path_style] (c3r3) edge node{} (c4r3);
\draw[path_style] (c4r3) edge node{} (c5r3);
\Edge (c5r3)(c6r3);
\Edge (c6r3)(c7r3);
\draw[dashed_style] (c0r4) edge node{} (c1r4);
% \Edge (c0r4)(c1r4);
\Edge (c1r4)(c2r4);
\Edge (c2r4)(c3r4);
\Edge (c3r4)(c4r4);
\Edge (c4r4)(c5r4);
\Edge (c5r4)(c6r4);
\draw[path_style] (c6r4) edge node{} (c7r4);
\draw[dashed_style] (c0r5) edge node{} (c1r5);
% \Edge (c0r5)(c1r5);
\Edge (c1r5)(c2r5);
\Edge (c2r5)(c3r5);
\Edge (c3r5)(c4r5);
\Edge (c4r5)(c5r5);
\Edge (c5r5)(c6r5);
\Edge (c6r5)(c7r5);
%vertical edges
% \draw (c2r4) edge node{} (c3r5);
\draw[dashed_style] (c0r0) -- (c0r1)  node[midway,left] {c};
\Edge (c1r0)(c1r1);
\Edge (c2r0)(c2r1);
\Edge (c3r0)(c3r1);
\Edge (c4r0)(c4r1);
\Edge (c5r0)(c5r1);
\Edge (c6r0)(c6r1);
\Edge (c7r0)(c7r1);
\draw[dashed_style] (c0r1) -- (c0r2)  node[midway,left] {a};
\Edge (c1r1)(c1r2);
\Edge (c2r1)(c2r2);
\Edge (c3r1)(c3r2);
\Edge (c4r1)(c4r2);
\Edge (c5r1)(c5r2);
\Edge (c6r1)(c6r2);
\Edge (c7r1)(c7r2);
\draw[dashed_style] (c0r2) -- (c0r3)  node[midway,left] {c};
\Edge (c1r2)(c1r3);
\Edge (c2r2)(c2r3);
\Edge (c3r2)(c3r3);
\Edge (c4r2)(c4r3);
\Edge (c5r2)(c5r3);
\Edge (c6r2)(c6r3);
\Edge (c7r2)(c7r3);
\draw[dashed_style] (c0r3) -- (c0r4)  node[midway,left] {a};
\Edge (c1r3)(c1r4);
\draw[->] (c2r3) edge node{} (c2r4);
\Edge (c3r3)(c3r4);
\Edge (c4r3)(c4r4);
\Edge (c5r3)(c5r4);
\Edge (c6r3)(c6r4);
\Edge (c7r3)(c7r4);
\draw[dashed_style] (c0r4) -- (c0r5)  node[midway,left] {b};
\Edge (c1r4)(c1r5);
\Edge (c2r4)(c2r5);
\Edge (c3r4)(c3r5);
\Edge (c4r4)(c4r5);
\Edge (c5r4)(c5r5);
\Edge (c6r4)(c6r5);
\Edge (c7r4)(c7r5);
%diagonal edges
\draw[dashed_style] (c0r0) edge node{} (c1r1);
% \Edge (c0r0)(c1r1);
\Edge (c1r0)(c2r1);
\Edge (c3r0)(c4r1);
\Edge (c4r0)(c5r1);

\Edge (c5r1)(c6r2);
\Edge (c6r1)(c7r2);

\draw[dashed_style] (c0r2) edge node{} (c1r3);
% \Edge (c0r2)(c1r3);
\draw[path_style] (c1r2) edge node{} (c2r3);
\Edge (c3r2)(c4r3);
\Edge (c4r2)(c5r3);

\draw[path_style] (c5r3) edge node{} (c6r4);
\Edge (c6r3)(c7r4);

\draw[->] (c2r4) edge node{} (c3r5);
% \draw[dashed_style] (c2r4) edge node{} (c3r5);

\end{tikzpicture}

%% file: figure_k-iams-psam-new.tex
\begin{tikzpicture}[scale=0.65, transform
shape, dashed_style/.style={->,dashed}, path_style/.style={->,very thick}]

\tikzstyle{greyvertex}=[circle, draw, inner sep=0pt, minimum
size=6pt, fill=black!50]
\tikzstyle{lightgreyvertex}=[circle, draw, inner sep=0pt, minimum
size=6pt, fill=light gray]

    \draw[thick, ->] (-0.5,0.5) -- (-0.5,-1) node[midway,left] {$\bstring$};
    \draw[thick, ->] (-0.5,0.5) --  (1,0.5) node[midway,above] {$\astring$};
\node[greyvertex] (c0r0) at (0,-0) {};
\vertex (c1r0) at (1,-0) {};
\vertex (c2r0) at (2,-0) {};
\vertex (c3r0) at (3,-0) {};
\vertex (c4r0) at (4,-0) {};
\vertex (c5r0) at (5,-0) {};
\vertex (c6r0) at (6,-0) {};
\vertex (c7r0) at (7,-0) {};
\node[greyvertex] (c0r1) at (0,-1) {};
\vertex (c1r1) at (1,-1) {};
\vertex (c2r1) at (2,-1) {};
\vertex (c3r1) at (3,-1) {};
\vertex (c4r1) at (4,-1) {};
\vertex (c5r1) at (5,-1) {};
\vertex (c6r1) at (6,-1) {};
\vertex (c7r1) at (7,-1) {};
\node[greyvertex] (c0r2) at (0,-2) {};
\vertex (c1r2) at (1,-2) {};
\vertex (c2r2) at (2,-2) {};
\vertex (c3r2) at (3,-2) {};
\vertex (c4r2) at (4,-2) {};
\vertex (c5r2) at (5,-2) {};
\vertex (c6r2) at (6,-2) {};
\vertex (c7r2) at (7,-2) {};
\node[greyvertex] (c0r3) at (0,-3) {};
\vertex (c1r3) at (1,-3) {};
\vertex (c2r3) at (2,-3) {};
\vertex (c3r3) at (3,-3) {};
\vertex (c4r3) at (4,-3) {};
\vertex (c5r3) at (5,-3) {};
\vertex (c6r3) at (6,-3) {};
\vertex (c7r3) at (7,-3) {};
\node[greyvertex] (c0r4) at (0,-4) {};
\node[greyvertex] (c1r4) at (1,-4) {};
\node[greyvertex] (c2r4) at (2,-4) {};
\node[greyvertex] (c3r4) at (3,-4) {};
\node[greyvertex] (c4r4) at (4,-4) {};
\node[greyvertex] (c5r4) at (5,-4) {};
\node[greyvertex] (c6r4) at (6,-4) {};
\node[greyvertex] (c7r4) at (7,-4) {};
\vertex (c0r5) at (0,-5) {};
\vertex (c1r5) at (1,-5) {};
\vertex (c2r5) at (2,-5) {};
\vertex (c3r5) at (3,-5) {};
\vertex (c4r5) at (4,-5) {};
\vertex (c5r5) at (5,-5) {};
\vertex (c6r5) at (6,-5) {};
\vertex (c7r5) at (7,-5) {};
\tikzset{EdgeStyle/.style={->}}
%horizontal edges
\draw[->] (c0r0) -- (c1r0)  node[midway,above] {c};
\draw[->] (c1r0) -- (c2r0)  node[midway,above] {c};
\draw[->] (c2r0) -- (c3r0)  node[midway,above] {b};
\draw[->] (c3r0) -- (c4r0)  node[midway,above] {c};
\draw[->] (c4r0) -- (c5r0)  node[midway,above] {c};
\draw[->] (c5r0) -- (c6r0)  node[midway,above] {a};
\draw[->] (c6r0) -- (c7r0)  node[midway,above] {a};
\draw[->] (c0r1) edge node{} (c1r1);
\Edge (c1r1)(c2r1);
\Edge (c2r1)(c3r1);
\Edge (c3r1)(c4r1);
\Edge (c4r1)(c5r1);
\Edge (c5r1)(c6r1);
\Edge (c6r1)(c7r1);
\Edge (c0r2)(c1r2);
\Edge (c1r2)(c2r2);
\Edge (c2r2)(c3r2);
\Edge (c3r2)(c4r2);
\Edge (c4r2)(c5r2);
\Edge (c5r2)(c6r2);
\Edge (c6r2)(c7r2);
\Edge (c0r3)(c1r3);
\Edge (c1r3)(c2r3);
\Edge (c2r3)(c3r3);
\Edge (c3r3)(c4r3);
\Edge (c4r3)(c5r3);
\Edge (c5r3)(c6r3);
\Edge (c6r3)(c7r3);
\Edge (c0r4)(c1r4);
\Edge (c1r4)(c2r4);
\Edge (c2r4)(c3r4);
\Edge (c3r4)(c4r4);
\Edge (c4r4)(c5r4);
\Edge (c5r4)(c6r4);
\Edge (c6r4)(c7r4);
\draw[dashed_style] (c0r5) edge node{} (c1r5);
\draw[dashed_style] (c1r5) edge node{} (c2r5);
\draw[dashed_style] (c2r5) edge node{} (c3r5);
\draw[dashed_style] (c3r5) edge node{} (c4r5);
\draw[dashed_style] (c4r5) edge node{} (c5r5);
\draw[dashed_style] (c5r5) edge node{} (c6r5);
\draw[dashed_style] (c6r5) edge node{} (c7r5);
%vertical edges
\draw[->] (c0r0) -- (c0r1)  node[midway,left] {c};
\draw[->] (c1r0) edge node{} (c1r1);
\draw[->] (c2r0) edge node{} (c2r1);
\draw[->] (c3r0) edge node{} (c3r1);
\draw[->] (c4r0) edge node{} (c4r1);
\draw[->] (c5r0) edge node{} (c5r1);
\draw[->] (c6r0) edge node{} (c6r1);
\draw[->] (c7r0) edge node{} (c7r1);
\draw[->] (c0r1) -- (c0r2)  node[midway,left] {a};
\Edge (c1r1)(c1r2);
\Edge (c2r1)(c2r2);
\Edge (c3r1)(c3r2);
\Edge (c4r1)(c4r2);
\Edge (c5r1)(c5r2);
\Edge (c6r1)(c6r2);
\Edge (c7r1)(c7r2);
\draw[->] (c0r2) -- (c0r3)  node[midway,left] {c};
\Edge (c1r2)(c1r3);
\Edge (c2r2)(c2r3);
\Edge (c3r2)(c3r3);
\Edge (c4r2)(c4r3);
\Edge (c5r2)(c5r3);
\Edge (c6r2)(c6r3);
\Edge (c7r2)(c7r3);
\draw[->] (c0r3) -- (c0r4)  node[midway,left] {a};
\Edge (c1r3)(c1r4);
\Edge (c2r3)(c2r4);
\Edge (c3r3)(c3r4);
\Edge (c4r3)(c4r4);
\Edge (c5r3)(c5r4);
\Edge (c6r3)(c6r4);
\Edge (c7r3)(c7r4);
\draw[dashed_style] (c0r4) -- (c0r5)  node[midway,left] {b};
\draw[dashed_style] (c1r4) edge node{} (c1r5);
\draw[dashed_style] (c2r4) edge node{} (c2r5);
\draw[dashed_style] (c3r4) edge node{} (c3r5);
\draw[dashed_style] (c4r4) edge node{} (c4r5);
\draw[dashed_style] (c5r4) edge node{} (c5r5);
\draw[dashed_style] (c6r4) edge node{} (c6r5);
\draw[dashed_style] (c7r4) edge node{} (c7r5);
%diagonal edges
\draw[->]  (c0r0) edge node{} (c1r1);
\draw[->]  (c1r0) edge node{} (c2r1);
\draw[->]  (c3r0) edge node{} (c4r1);
\draw[->]  (c4r0) edge node{} (c5r1);

\Edge (c5r1)(c6r2);
\Edge (c6r1)(c7r2);

\Edge (c0r2)(c1r3);
\Edge (c1r2)(c2r3);
\Edge (c3r2)(c4r3);
\Edge (c4r2)(c5r3);

\Edge (c5r3)(c6r4);
\Edge (c6r3)(c7r4);

\draw[dashed_style] (c2r4) edge node{} (c3r5);

\end{tikzpicture}

%% file: k-ssam-algorithm-example.tex
\begin{tikzpicture}[scale=0.5]

\tikzstyle{pivotalpoint}=[circle, draw, inner sep=0pt, minimum
size=5pt, fill=gray]
\tikzstyle{matchpoint}=[red]

\usetikzlibrary{decorations.pathreplacing}

\draw[thick] (0,0) -- (10,0) -- (10,10) -- (0,10) -- (0,0);

\draw[gray,very thin] (0,0) grid (10,10);

%pivotal points
\node[pivotalpoint] (p1) at (3,9) {};
\node[pivotalpoint] (p2) at (5,8) {};
\node[pivotalpoint] (p3) at (4,7) {};
\node[pivotalpoint] (p4) at (7,5) {};
\node[pivotalpoint] (p5) at (8,4) {};
\node[pivotalpoint] (p6) at (9,2) {};

%match points
\draw[matchpoint] (0,7) -- (10,7);
\draw[matchpoint] (0,5) -- (10,5);
\draw[matchpoint] (0,2) -- (10,2);
\draw[matchpoint] (0,1) -- (10,1);

\draw[decorate,decoration={brace,amplitude=5pt}]
(-0.5,7) -- (-0.5,9) node [black,midway,xshift=-0.4cm] 
{\footnotesize $b_1$};

\draw[decorate,decoration={brace,amplitude=3pt}]
(-0.5,4) -- (-0.5,5) node [black,midway,xshift=-0.4cm] 
{\footnotesize $b_2$};

\draw (-0.5,2) -- (-0.5,2) node [midway,xshift=-0.4cm] 
{\footnotesize $b_3$};

% \draw[fill=yellow] (1,1) circle (2.5pt);
\end{tikzpicture}

%% file: k-psam-algorithm-example.tex
\begin{tikzpicture}[scale=0.5]

\tikzstyle{pivotalpoint}=[circle, draw, inner sep=0pt, minimum
size=5pt, fill=gray]
\tikzstyle{matchpoint}=[red]

\usetikzlibrary{decorations.pathreplacing}

\draw[thick] (0,0) -- (9,0) -- (9,8) -- (0,8) -- (0,0);

\draw[gray,very thin] (0,0) grid (9,8);

%pivotal points
\node[pivotalpoint] (p1) at (5,5) {};
\node[pivotalpoint] (p2) at (2,3) {};
\node[pivotalpoint] (p3) at (1,2) {};
\node[pivotalpoint] (p4) at (6,1) {};

%match points
\draw[matchpoint] (0,7) -- (9,7);
\draw[matchpoint] (0,6) -- (9,6);
\draw[matchpoint] (0,4) -- (9,4);
\draw[matchpoint] (0,1) -- (9,1);

% \draw[decorate,decoration={brace,amplitude=5pt}]
% (-0.5,7) -- (-0.5,9) node [black,midway,xshift=-0.4cm] 
% {\footnotesize $b_1$};
% 
% \draw[decorate,decoration={brace,amplitude=3pt}]
% (-0.5,4) -- (-0.5,5) node [black,midway,xshift=-0.4cm] 
% {\footnotesize $b_2$};
% 
% 
% \draw (-0.5,2) -- (-0.5,2) node [midway,xshift=-0.4cm] 
% {\footnotesize $b_3$};

% \draw[fill=yellow] (1,1) circle (2.5pt);
\end{tikzpicture}